\documentclass[11 pt]{article}

\usepackage[left=1in, right=1in, top=1in, bottom=1in]{geometry}
\usepackage[group-separator={,}]{siunitx}
\usepackage{amsmath,amsfonts,amssymb, amsthm}
\usepackage{graphicx}
\usepackage{url}
\usepackage{mathtools}
\usepackage[textsize=tiny,textwidth=1cm,shadow]{todonotes}
\usepackage{thmtools}
\usepackage{enumitem}
\usepackage{pdflscape}
\usepackage{verbatim}
\usepackage{nccmath}
\usepackage[utf8]{inputenc}
\usepackage{array}
\usepackage{arydshln}
\usepackage{multirow}
\usepackage{bigdelim}
\usepackage{subcaption}
\definecolor{MyBlue}{rgb}{0.12, 0.12, 0.76}
\usepackage[colorlinks,allcolors=MyBlue]{hyperref}
\usepackage{float}
\usepackage{pgfplots}
\usepackage{xspace}
\usepackage[numbers]{natbib}
\usepackage[font=small]{caption}

\usepackage{algorithm}
\usepackage{algpseudocode}
\usepackage{algorithmicx}
\let\oldReturn\Return
\renewcommand{\Return}{\State\oldReturn}
\algtext*{EndWhile}
\algtext*{EndIf}
\algtext*{EndForAll}
\algtext*{EndFor}
\algtext*{EndFunction}

\usepackage{pgf}
\usepackage{tikz}
\usetikzlibrary{shadows,arrows,decorations,decorations.shapes,backgrounds,shapes,snakes,automata,fit,petri,shapes.multipart,calc,positioning,shapes.geometric,graphs,graphs.standard,plotmarks}

\usepackage{newfloat}
\DeclareFloatingEnvironment[
    fileext=los,
    listname={List of Examples},
    name=Example,
    placement=Htbhp,
]{example}

\makeatletter
\newcommand{\thickhline}{%
    \noalign {\ifnum 0=`}\fi \hrule height 1.4pt
    \futurelet \reserved@a \@xhline
}
\newcolumntype{"}{@{\hskip\tabcolsep\vrule width 1.4pt\hskip\tabcolsep}}
\makeatother

\allowdisplaybreaks 

\usepackage{thmtools}
\usepackage{thm-restate}
\declaretheorem[name=Theorem,numberwithin=section]{thm}
\newtheorem{theorem}[thm]{Theorem}
\newtheorem{lemma}[theorem]{Lemma}
\newtheorem{corollary}{Corollary}[theorem]
\newtheorem{definition}{Definition}[section]
\theoremstyle{remark}
\newtheorem*{remark}{Remark}

\DeclareMathOperator*{\argmax}{arg\,max}
\DeclareMathOperator*{\argmin}{arg\,min}

\newcommand{\x}{\mathbf{x}}
\newcommand{\y}{\mathbf{y}}
\newcommand{\z}{\mathbf{z}}
\newcommand{\bfu}{\mathbf{u}}

\newcommand{\p}{\mathbf{p}}
\newcommand{\f}{\mathbf{f}}
\newcommand{\g}{\mathbf{g}}
\newcommand{\A}{\mathbf{a}}
\newcommand{\B}{\mathbf{b}}

\newcommand\bbr{\mathbb{R}}
\newcommand\ep{\varepsilon}
\newcommand\lam{\lambda}
\newcommand\bflam{\boldsymbol{\lam}}
\newcommand\q{\mathbf{q}}

\newcommand\bbrpos{\mathbb{R}_{\geq 0}}

\newcommand\xs{\mathbf{x^*}}
\newcommand\aprime{\mathbf{a'}}
\newcommand\bprime{\mathbf{b'}}

\newcommand\tat{t\^{a}tonnement\xspace}

\begin{document}

\title{Markets Beyond Nash Welfare for Leontief Utilities}
\author{Ashish Goel \and Reyna Hulett \and Benjamin Plaut}
\date{\{ashishg, rmhulett, bplaut\}@stanford.edu}

\maketitle


\begin{abstract}
\noindent We study the allocation of divisible goods to competing agents via a market mechanism, focusing on agents with Leontief utilities. The majority of the economics and mechanism design literature has focused on \emph{linear} prices, meaning that the cost of a good is proportional to the quantity purchased. Equilibria for linear prices are known to be exactly the maximum Nash welfare allocations.

\emph{Price curves} allow the cost of a good to be any (increasing) function of the quantity purchased. First, we show that an allocation can be supported by strictly increasing price curves if and only if it is \emph{group-domination-free}. A similar characterization holds for weakly increasing price curves. We use this to show that given any allocation, we can compute strictly (or weakly) increasing price curves that support it (or show that none exist) in polynomial time. These results use a variant of Farkas' Lemma along with a combinatorial argument to construct piecewise linear price curves. For our second main result, we use Lagrangian duality to show that in the bandwidth allocation setting, any allocation maximizing a CES welfare function can be supported by price curves. Taken together, our results show that nonlinear pricing opens up multiple possibilities beyond Nash welfare for market equilibria.
\end{abstract}

\section{Introduction}\label{sec:intro}

In a market, buyers and sellers exchange goods according to some sort of pricing system. One of the most important concepts in the study of markets is \emph{market equilibrium}, which describes when the supply provided by the sellers and the demands of the buyers exactly match. Market equilibrium theory dates back to Walras's seminal work in 1874~\cite{walras_elements_1874}. In 1954, Arrow and Debreu finally showed that under some mild assumptions, a market equilibrium is guaranteed to exist for a wide class of utility functions~\cite{arrow_existence_1954}. This includes Leontief utilities, which will be our focus.

The simplest mathematical model of a market is a \emph{Fisher market}, first proposed in 1891 by Irving Fisher (see \cite{brainard_how_2005} for a modern introduction). A Fisher market consists of a set of goods available for sale, and a set of agents, each with a fixed amount of money to spend. It is usually assumed that agents have no value for leftover money. In Fisher markets, each good $j$ has a single real-number price $p_j$, and the cost of buying some quantity $x$ of good $j$ is $p_j \cdot x$. We refer to such prices as \emph{linear}, meaning that the cost is proportional to the quantity purchased. A market equilibrium assigns a price to each good such that when each agent purchases her favorite bundle of goods that is affordable under these prices, the demand exactly matches the supply. 

There are three motivations behind this work. First, in real market economies, prices are often not linear, and depend on the quantity purchased\footnote{One consequence of this is that there can be an incentive for agents to ``team up", which is not the case in linear pricing. For example, it could be cheaper for one person to purchase the resource in bulk and then distribute it, as opposed to each person buying her own: imagine ordering pizza for a party. We do not consider strategic behavior in this paper; see Section~\ref{sec:prior} for additional discussion.}. We refer to prices of this form as \emph{price curves}. For example, ``buying in bulk" may allow agents to purchase twice as much of some resource for less than twice the price. In this case, the marginal price decreases as more of the good is purchased. On the other hand, for a scarce resource, a central authority may choose to impose increasing marginal costs to ensure that no single individual can monopolize the resource. Israel's pricing policy for water is a good example of this, where each additional unit of water costs more than the previous one~\cite{becker_2015_water}. A tremendous amount of work has been devoted to understanding the nature of linear prices, despite the pervasiveness of price curves in the real world. This paper attempts to ask the same fundamental questions of price curves that have been answered for linear prices.

Second, imagine a social planner or mechanism designer who wishes to design a pricing scheme that will maximize some objective function. The objective function of a social planner is typically referred to as \emph{welfare}. There are many different social welfare functions, the most well studied being utilitarian welfare (the sum of agent utilities), Nash welfare (the product of agent utilities)~\cite{nash_bargaining_1950, kaneko_nash_1979}, and max-min welfare (the minimum agent utility~\cite{rawls_1971_theory, sen_1976_welfare, sen_1977_social})\footnote{The utilitarian welfare is also known as the Benthamite welfare, after Jeremy Bentham. The max-min welfare is also known as the Rawlsian welfare, after John Rawls, or the egalitarian welfare.}. Max-min welfare can be seen as caring only about equality across individuals. The utilitarian welfare measures overall good across the entire population, possibly at the expense of certain individuals. The Nash welfare is something of a compromise between these two extremes.

Eisenberg and Gale famously showed that for linear prices and a large class of agent utilities (including Leontief), the market equilibria correspond exactly to the allocations maximizing Nash welfare~\cite{eisenberg_aggregation_1961, eisenberg_consensus_1959}. This result is powerful, but also limiting: what if the social planner wishes to maximize a different welfare function? Is it possible that using price curves instead of linear prices allows a wider set of allocations to be equilibria? 
In particular, are there welfare functions other than Nash welfare such that welfare-maximizing allocations can always be supported by price curves? (We say that an allocation can be \emph{supported} by price curves if there exist prices curves that make that allocation an equilibrium.)  Our paper answers these questions in the affirmative.

The third motivation involves a more conceptual connection between markets and welfare functions, both of which have been extensively studied in the economics literature. We know that linear-pricing equilibria correspond to maximizing Nash welfare, but does this connection go deeper? Our work hints at an affirmative answer to this question as well.

\subsection{Bandwidth allocation}\label{sec:bandwidth-intro}
Resource allocation with Leontief utilities generalizes the problem of network bandwidth allocation, which is a well-studied area in its own right (for example, the work of Kelly~\cite{kelly_1998_rate} on proportional fairness). In bandwidth allocation, each agent wishes to transmit data along a fixed
route of links, and thus desires bandwidth for exactly those links in
equal amounts. In our setting, each link corresponds to a good, and
the agent has unweighted Leontief utility over the set of goods
corresponding to her desired route.

In the bandwidth allocation setting, price curves correspond naturally to a signaling mechanism that provides congestion signals (eg. in the form of a packet mark or drop) and an end-point protocol such as TCP~\cite{cerf_1974_protocol} corresponds naturally to agent responses. It has been known that different marking schemes (such as RED and CHOKe~\cite{floyd_1993_random, pan_2000_choke}) and versions of TCP lead to different objective functions~\cite{padhye_1998_modeling}, with CES welfare (also known as ``$\alpha$-fairness") being one such objective~\cite{Bonald2001, Mo2000}. However, a market-based understanding was developed only for Nash Welfare, starting with the seminal work of Kelly et al.~\cite{kelly_1998_rate}.

\subsection{CES welfare functions}

For any constant $\rho \in (-\infty, 0) \cup (0,1]$, the \emph{constant elasticity of substitution} (CES) welfare function is given by
\[
\Big(\sum\limits_{i \in N} u_i^{\rho}\Big)^{1/\rho}
\]
where $u_i$ is agent $i$'s utility. Setting $\rho = 1$ yields utilitarian welfare, and the limits as $\rho \to -\infty$ and $\rho \to 0$ yield max-min welfare and Nash welfare, respectively. The smaller $\rho$ is, the more the social planner cares about individual equality (max-min welfare being the extreme case of this), and the larger $\rho$ is, the more the social planner cares about overall societal good (utilitarian welfare being the extreme case of this). For this reason, $\rho$ is called the \emph{inequality aversion} parameter. 

This class of welfare functions was first proposed by Atkinson~\cite{atkinson_1970_measurement} and further developed by~\cite{blackorby_1978_measures}. The CES welfare function (as opposed to the CES agent utility function) has received very little attention in the computational economics community, despite being extremely influential in the traditional economics literature with~\cite{atkinson_1970_measurement} having over 8000 citations.



These welfare functions also admit an axiomatic characterization:
\begin{enumerate}
\item Monotonicity: if one agent's utility increases while all others are unchanged, the welfare function should prefer the new allocation.
\item Symmetry: the welfare function should treat all agents the same.
\item Continuity: the welfare function should be continuous.
\item Independence of common scale: scaling all agent utilities by the same factor should not affect which allocations have better welfare than others.
\item Independence of unconcerned agents: when comparing the welfare of two allocations, the comparison should not depend on agents who have the same utility in both allocations. 
\item The Pigou-Dalton principle: all things being equal, the welfare function should prefer more equitable allocations~\cite{dalton_1920_measurement,pigou_1912_wealth}. 
\end{enumerate}

Up to monotonic transformations of the welfare function (which of course do not affect which allocations have better welfare than others), the set of welfare functions that satisfy these axioms is exactly the set of CES welfare functions with $\rho \in (-\infty, 0)\cup(0,1]$,
including Nash welfare~\cite{moulin_2003_fair}\footnote{This actually does not include max-min welfare, which obeys weak monotonicity but not strict monotonicity.}. This axiomatic characterization shows that we are not just focusing on an arbitrary class of welfare functions: CES welfare functions are arguably the most reasonable welfare functions.

\section{Results and prior work}\label{sec:results}

We assume throughout the paper that agents have Leontief utility functions. An agent with a Leontief utility function desires the goods in fixed proportions, e.g., one unit of CPU for every two units of RAM. We can express agent $i$'s utility as
\[
\min_{j \in : w_{ij} \ne 0} \frac{x_{ij}}{w_{ij}}
\]
where $M$ is the set of goods, $x_{ij}$ is the amount of good $j$ agent $i$ receives, and $w_{ij}$ is agent $i$'s (nonnegative) weight for good $j$. The setting where $w_{ij} \in \{0,1\}$ for all $i$ and $j$ is also known as bandwidth allocation.

Leontief utilities exhibit certain convenient properties that other utility functions do not. In particular, such an agent will always purchase her goods exactly in the same proportions, and all that changes is how much she purchases. We also assume that each agent has the same amount of money to spend. However, most of our results do extend to the case of unequal budgets, as noted later on.


\subsection{Results}

\subsubsection{A necessary and sufficient condition for the existence of price curves.} Section~\ref{sec:gdf} presents our first main result, which concerns the first motivation described above: trying to understand fundamental properties of price curve equilibria. In particular, this section answers the following question: given some allocation, is there a way to tell whether there exist price curves that make this allocation an equilibrium? Furthermore, can such price curves be efficiently computed? 

The answer boils down to a property we call \emph{group-domination-freeness}. Roughly, a set of agents $\A$ group-dominates a set of agents $\B$ if these sets are the same size, but for every good $j$ and every threshold $\tau\in\bbrpos$, the number of agents in $\A$ receiving at least of $\tau$ of good $j$ is at least as large as the number of agents in $\B$ receiving at least $\tau$ of good $j$. The formal definition of group domination is given in Section~\ref{sec:gdf}. An allocation is group-domination-free (GDF) if no group dominates any other group. We show that an allocation can be supported by \emph{strictly} increasing price curves if and only if the allocation is GDF (Theorem~\ref{thm:gdf-final})\footnote{This result extends to the setting of unequal budgets if one instead considers ``budget-weighted group-domination-freeness". We elaborate on this in Section~\ref{sec:gdf}.}. This characterization
results in a polynomial time algorithm to compute the underlying price
curves or show that none exist (Theorem~\ref{thm:gdf-lp}). Appendix~\ref{sec:laf} gives an analogous characterization
theorem and polynomial time algorithm for weakly increasing price
curves (Theorems~\ref{thm:laf} and \ref{thm:laf-lp}).

Although the definition of group domination may seem slightly technical, we also demonstrate its relation to the concept of stochastic dominance, and argue that it may in fact be interpreted as a fairness notion. The stochastic dominance interpretation will also suggest that group domination is quite a strong property, and hence group-domination-freeness is a quite a weak assumption.

The proof of these characterization theorems involves the construction of a special matrix we call the \emph{agent-order matrix} $A$, which is a function of the allocation. We show that existence of strictly increasing price curves is captured by \emph{strongly} positive solutions (every entry of the solution vector is positive) to $A \mathbf{y} = \mathbf{0}$. We relate group-domination-freeness to a property of this matrix, and then invoke a duality theorem equivalent to Farkas' Lemma~\cite{perng_2017} to complete the proof. The algorithm for computing price curves is a linear program involving the agent-order matrix.

\subsubsection{Maximum CES welfare allocations can be supported in bandwidth allocation.} Our second main result concerns the second and third motivations: a social planner who wishes to maximize a welfare function other than Nash welfare, and understanding the connection between markets and welfare functions. We know that the maximum Nash welfare allocations can be supported by linear prices. If we allow price curves, are there other welfare functions whose maxima can be supported?

First, we will need some assumption on the agents' weights (recall that $w_{ij}$ denotes agent $i$' weight for good $j$). To see this, consider just two agents and one good. Since the agents have the same budget, they must receive equal amounts of the good no matter the price curve. However, if one agent derives less utility per unit of the good, this allocation doesn't maximize any CES welfare function except for Nash welfare\footnote{This example actually holds for a much wider class of utilities, not just Leontief. This is because for a single good, all anyone can do is buy as much of that good as they can.}. One natural way to handle this is to assume that the agents' weight vectors are normalized in some manner. The bandwidth allocation setting -- $w_{ij} \in \{0,1\}$ for all $i$ and $j$ -- constitutes one such possibility (refer to Section~\ref{sec:bandwidth-intro} for additional discussion of this setting).

Our second main result is that in the bandwidth allocation setting, the welfare-maximizing allocations for any fixed CES welfare function with $\rho \in (-\infty, 0)\cup(0, 1)$ can be supported by price curves (Theorem~\ref{thm:bandwidth}). We prove this by writing a convex program to maximize CES welfare, and using duality to construct explicit price curves. Furthermore, these price curves take on a natural form: the cost of buying $x$ of good $j$ is $q_j x^{1-\rho}$, where $q_j \ge 0$ is a constant derived from the dual\footnote{These results extend to agents with unequal budgets if we instead consider the ``budget-weighted CES welfare", i.e., $\big(\sum_{i \in N} B_i u_i^\rho\big)^{1/\rho}$, where $B_i$ is agent $i$'s budget. We discuss this in Section~\ref{sec:ces-unequal-budgets}. The price curves will take the exact same form.}. This result can be thought of as extending the work on price-based
congestion control (pioneered by Kelly et al.~\cite{kelly_1998_rate}) beyond Nash welfare to almost all CES welfare functions.

We also prove a converse of sorts: if an allocation $\x$ can be supported by price curves of the form $q_j x^{1-\rho}$, and the supply is exhausted for every good with nonzero price (i.e., $q_j \ne 0$), then $\x$ is a maximum CES welfare allocation (Theorem~\ref{thm:ces-converse}). This is analogous to the famous result of Eisenberg and Gale: the linear-pricing equilibrium allocations are exactly the allocations maximizing Nash welfare~\cite{eisenberg_aggregation_1961, eisenberg_consensus_1959}.

One may wonder if Theorem~\ref{thm:bandwidth} could be extended to $\rho = 1$, i.e., maximizing the sum of utilities. Example~\ref{ex:util} shows that the answer is no, unfortunately. One may also wonder if Theorem~\ref{thm:bandwidth} would generalize if we relax our constraint from $w_{ij} \in \{0,1\}$ to $w_{ij} \in [0,1]$. The answer is again no; this counterexample is more involved and is given by Theorem~\ref{thm:L-infty-counter} in Section~\ref{sec:counterexamples}.

\begin{example}[tb]
\centering
\begin{tabular}{ c|ccc} 
& agent 1 & agent 2 & agent 3\\
\hline
 good 1 & 1 & 0 & 1  \\
good 2 & 0 & 1 & 1
\end{tabular}
\caption{A bandwidth allocation instance where no maximum utilitarian welfare allocation can be supported. The table above gives each agent's weight $w_{ij} \in \{0,1\}$ for each good. Utilitarian welfare is maximized by giving all of good 1 to agent 1 and all of good 2 to agent 2, leaving agent 3 with nothing. This is impossible to support with price curves, because agent 3 can always buy a nonzero amount of the goods she wants.}
\label{ex:util}
\end{example}

\subsubsection{Information required by the social planner.} In general, the price curves will depend on agents' preferences, and so the social planner needs to know agents' preferences in order to compute them. This is true of linear-pricing markets as well: the equilibrium prices depend on the utility functions of the agents. For our GDF characterization result, the price curves can have a very complex shape that depends intricately on the specific preferences, unlike linear prices. For this reason, we view the GDF characterization more as a conceptual contribution than an actual mechanism. In contrast, for our CES welfare bandwidth allocation result, the price curves have a very simple shape that is independent of the agents' utility functions (the price of buying $x$ will be $q_j x^{1-\rho}$, where $q_j$ is a Lagrange multiplier corresponding to good $j$). This structure suggests a simple decentralized primal-dual algorithm similar to the work of Kelly et al.~\cite{kelly_1998_rate}, where on each step, every agent updates her (primal) allocation based on the current prices on the links she cares about, and each link updates its (dual) price based on the total flow through that link. We discuss this further in Section~\ref{sec:decentralized}.

\subsubsection{Additional results.} We prove two additional results. First, we consider max-min welfare in Appendix~\ref{sec:maxmin}, and show that as long as agents' weights are reasonably normalized, allocations with optimal max-min welfare can be supported. Second, in Appendix~\ref{sec:laf}, we give a characterization theorem for weakly increasing price curves (Theorem~\ref{thm:laf}). Just like strictly increasing price curves, we can compute weakly increasing price curves (or show that none exist) in polynomial time, using a linear program (Theorem~\ref{thm:laf-lp}).

\subsection{Prior work}\label{sec:prior}

The study of markets has a long history
in the economics literature~\cite{walras_elements_1874, varian_equity_1974,arrow_existence_1954, brainard_how_2005}. Recently, this topic has received significant attention in the computer science community as well (see~\cite{vazirani_2007_combinatorial} for an algorithmic exposition). The vast majority of the literature has focused on linear prices. Perhaps the most relevant classical result is the Second Welfare Theorem, which states that for any Pareto optimal allocation, there exists a (possibly unequal) redistribution of initial wealth which makes that allocation a (linear-pricing) market equilibrium. In a similar vein, \cite{foley_lindahls_1970} showed that for linear but fully personalized prices (i.e., we can independently assign different prices to different agents for the same good), one can support any Pareto optimal allocation.

The important question, then, is what price curve equilibria offer that these prior results do not. First, in many of societies, a centrally mandated redistribution of wealth is out of the question. Similarly, fully personalized prices mean that we lose any claim to fairness, since agents may be subjected to totally different prices for the same resource. In contrast, price curves do not require a redistribution of wealth, and furthermore are \emph{anonymous}, meaning that all agents are subject to the same pricing scheme. These properties suggest that price curves are more practical, and indeed price curves do appear in practice (as noted previously, see Israel's pricing policy for water~\cite{becker_2015_water}).

The concept of a social welfare function -- a function which encapsulates a societal value system -- was first proposed in 1938~\cite{bergson_1938_reformulation}, and further developed by~\cite{samuelson_1947_foundations}. For brevity, we will just refer to them as welfare functions. Since then, several different welfare functions have been proposed, the most well-studied being utilitarian welfare, Nash welfare~\cite{nash_bargaining_1950, kaneko_nash_1979}, and max-min welfare (the minimum agent utility)~\cite{rawls_1971_theory, sen_1976_welfare, sen_1977_social}. The class of CES welfare functions was first proposed by~\cite{atkinson_1970_measurement} and further developed by \cite{blackorby_1978_measures}, although not under the same name. See \cite{moulin_2003_fair} for a modern discussion of welfare functions.

We briefly mention an important property in mechanism design: \emph{strategy-proofness}. A mechanism is strategy-proof if agents can never improve their utility by lying about their preferences. Unfortunately, even in simple settings, the only mechanism for resource allocation that can simultaneously guarantee strategy-proofness and Pareto optimality is \emph{dictatorial}, meaning that one agent receives all of the resources~\cite{schummer_strategy_1996}. This is clearly unacceptable, so we sacrifice strategy-proofness in favor of Pareto optimality. Specifically, we assume throughout the paper that agents always truthfully report their preferences.

The remainder of the paper is organized as follows. Section~\ref{sec:model} formally defines the model. Section~\ref{sec:gdf} presents our first main result: that for strictly increasing price curves, an allocation can be supported if and only if it is group-domination-free. Section~\ref{sec:ces-duality} gives our second main result: that in the bandwidth allocation setting, every maximum CES welfare allocation can be supported by price curves. At this point we conclude the main paper, and move on to supplementary results. Appendix~\ref{sec:maxmin} shows that allocations with optimal max-min welfare can be supported by price curves in a wide range of settings. In Appendix~\ref{sec:laf}, we generalize our characterization theorem from Section~\ref{sec:gdf} to account for weakly increasing price curves. Appendix~\ref{sec:counterexamples} provides counterexamples to various claims that one might have hoped to prove. We also discuss in that section why certain other classes of utilities -- in particular, linear utilities -- are much more difficult to analyze. Finally, Appendix~\ref{sec:extra-proofs} provides some proofs that are omitted from earlier sections.

\section{Model}\label{sec:model}

Let $N = \{1,2,\ldots ,n\}$ be the set of agents, and let $M$ be a set of divisible goods, where $|M| = m$. Throughout the paper, we use $i$ and $k$ to refer to agents and $j$ and $\ell$ to refer to goods. Let $s_j$ be the available supply of good $j$.\footnote{It is common to normalize the supply of each good to 1, but this is not possible when weights are restricted to be either 0 or 1.} The social planner needs to determine an \emph{allocation} $\x \in \bbr_{\ge 0}^{n\times m}$, where $x_i \in \bbr_{\ge 0}^m$ is the \emph{bundle} of agent $i$, and $x_{ij} \in [0,s_j]$ is the fraction of good $j$ allocated to agent $i$. An allocation cannot allocate more than the available supply: $\x$ is a valid allocation if and only if $\sum_i x_{ij} \leq s_j$ for all $j$.

Agent $i$'s utility for a bundle $x_i$ is denoted by $u_i(x_i)  \in \bbr_{\ge 0}$. The literature studies many subclasses of utility functions. For example, \emph{Leontief} utilities have the form
\[
u_i(x_i) = \min_{j \in M: w_{ij} \ne 0}\   \frac{x_{ij}}{w_{ij}}
\]
where $w_{ij}$ represents the weight agent $i$ has for good $j$. For brevity, we will usually write $u_i(x_i) = \min_{j \in M}  \frac{x_{ij}}{w_{ij}}$ and leave the $w_{ij} \ne 0$ condition implicit. The same holds for other contexts where we are dividing by a value $w_{ij}$ that may be zero. We assume that agents have Leontief utilities throughout the paper, and we assume that each agent has nonzero weight for at least one good.

Just as agents have utilities over the bundles they receive, we can imagine a social planner who wishes to design a pricing mechanism to maximize some societal welfare function $\Phi(\x)$. One can think of $\Phi$ as the social planner's utility function, which takes as input the agent utilities, instead of a bundle of goods. The most well-studied welfare functions are the \emph{max-min welfare} $\Phi(\x) = \min_{i \in N} u_i(x_i)$, the \emph{Nash welfare} $\Phi(\x) = \big(\prod_{i \in N} u_i(x_i) \big)^{1/n}$, and the \emph{utilitarian welfare} $\Phi(\x) = \sum_{i \in N} u_i(x_i)$.\footnote{One could also imagine the (arguably less natural) case of a social planner who cares about some agents' utilities more than others, which would manifest as weights in her utility function. We briefly consider this case in the setting of CES welfare with unequal budgets (Section~\ref{sec:ces-unequal-budgets}).} These three welfare functions can be generalized by a CES welfare function:
\[
\Phi(\x) = \Big(\sum\limits_{i \in N} u_i(x_i)^{\rho}\Big)^{1/\rho}
\]
where $\rho$ is a constant in $(-\infty, 0) \cup (0, 1]$.

\subsection{Fisher markets}\label{sec:fisher-markets}

The simplest market model is a \emph{Fisher market}~\cite{brainard_how_2005}. In a Fisher market, each good is available for sale and each agent enters the market with a fixed budget she wishes to spend. It is assumed that agents have no value for leftover money; this will imply that each agent always spends her entire budget. Unless otherwise stated, we will assume that all agents have the same budget, and normalize all budgets to 1 without loss of generality.

Informally, a \emph{Fisher market equilibrium} assigns a price $p_j \in \bbr_{\geq 0}$ to each good $j$ so that the agents' demand equals supply. Formally, for a price vector $\p \in \bbr^m_{\geq 0}$, the \emph{cost} of bundle $x_i$ is $C_{\p}(x_i) = \sum_{j\in M} x_{ij}p_j$. Bundle $x_i$ is \emph{affordable} for agent $i$ if $C_{\p}(x_i) \leq 1$. Agent $i$'s \emph{demand set} is the set of her favorite affordable bundles, i.e., 
\[
D_i(\p) = \argmax\limits_{x_i \in \bbr^m_{\geq 0}:\ C_\p(x_i) \leq 1} u_i(x_i).
\]

If $p_j > 0$ for all $j \in M$, an agent with Leontief utility will always purchase in exact proportion to her weights: since agent $i$'s utility is determined by $\min_{j \in M} \frac{x_{ij}}{w_{ij}}$, violating these proportions cannot increase her utility. Thus when discussing an arbitrary allocation $\x$, we assume that each bundle $x_i$ is in proportion to agent $i$'s weights: otherwise there is no hope of supporting such an allocation. For brevity, we leave this assumption implicit through the paper, rather than always stating ``for an arbitrary allocation $\x$ where each bundle is in proportion to agent $i$'s weights".

The careful reader may note that we are glossing over a detail: if $p_j = 0$ for some good $j$, agent $i$ can add more of good $j$ to her bundle at no additional cost. This does not affect agent utilities at all, but is technically possible. In order to avoid handling this uninteresting and sometimes messy edge case, we assume throughout the paper that for agents with Leontief utilities, demand sets and arbitrary allocations are always in exact proportion to agent weights.

Formally, a Fisher market equilibrium $(\x, \p)$ is an allocation $\x$ and price vector $\p \in \bbr^m_{\geq 0}$ such that
\begin{enumerate}
\item Each agent receives a bundle in her demand set: $x_i \in D_i(\p)$.
\item The market clears: for all $j \in M$, $\sum_{i \in N} x_{ij} \leq s_j$, and if $p_j > 0$, then $\sum_{i \in N} x_{ij} = s_j$. 
\end{enumerate}
When all agents have the same budget, this is also called the \emph{competitive equilibrium from equal incomes}~\cite{varian_equity_1974}. 


For a wide class of agent utilities, including Leontief utilities, an equilibrium is guaranteed to exist~\cite{arrow_existence_1954}\footnote{Specifically, an equilibrium is guaranteed to exist as long agent utilities are continuous, quasi-concave, and non-satiated. The full Arrow-Debreu model also allows for agents to enter to market with goods themselves and not only money; the necessary conditions on utilities are slightly more complex in that setting.}. Furthermore, the equilibrium allocations are the exactly the allocations which maximize Nash welfare\footnote{The conditions for the correspondence between Fisher market equilibria and Nash welfare are slightly stricter than those for market equilibrium existence, but are still quite general. Sufficient criteria were given in \cite{eisenberg_aggregation_1961} and generalized slightly by \cite{jain_eisenberggale_2010}.}. This is made explicit by the celebrated Eisenberg-Gale convex program~\cite{eisenberg_aggregation_1961, eisenberg_consensus_1959},
and combinatorial approaches to computing market
equilibria~\cite{devanur_2008_market,jain_eisenberggale_2010}.

\subsection{Price curves}\label{sec:price-curves}

Our paper considers an expanded model where instead of assigning a single price $p_j \in \bbr_{\geq 0}$ to each good, we assign each good a \emph{price curve} $f_j: \bbr_{\geq 0} \to \bbr_{\geq 0}$. The function $f_j$ expresses the cost of good $j$ as a function of the quantity purchased. When we say ``price curve", we mean a function $f_j$ that is weakly increasing (buying more of a good cannot cost less), normalized ($f_j(0) = 0$), and continuous. Setting $f_j(x) = p_j \cdot x$ for all $j \in M$ and all $x \in \bbr_{\geq 0}$ yields the Fisher market setting.

Given a vector of price curves $\f = (f_1,\dots,f_m)$, the cost of a bundle $x_i$ is now $C_\f(x_i) = \sum_{j \in M} f_j(x_{ij})$. Although the functions $f_j$ may not be linear, the cost of a bundle is still additive across goods. Each agent's demand set is defined identically to the Fisher market setting: $D_i(\f) = \argmax\limits_{x_i \in \bbr^m_{\geq 0}:\ C_\f(x_i) \leq 1} u_i(x_i)$.

The demand set is intuitively the same as in the Fisher market setting: each agent purchases exactly in proportion to her weights, and buys as much as she can without exceeding her budget. A \emph{price curve equilibrium} $(\x, \f)$ is an allocation $\x$ and vector of price curves $\f$ such that
\begin{enumerate}
\item Each agent receives a bundle in her demand set: $x_i \in D_i(\f)$.
\item The demand does not exceed supply: $\sum_{i \in N} x_{ij} \leq s_j$ for all $j \in M$\footnote{For Fisher market equilibria, the second condition also stipulated that whenever $p_j > 0$, $\sum_{j \in M} x_{ij} = s_j$. Without this additional condition, cranking up all prices to infinity would result in trivial equilibria where all agents purchase almost nothing and so would certainly not maximize Nash welfare. Such trivial price curve equilibria do exist under our definition, but since we are not going to make any claims of the form ``all price curve equilibria maximize a certain function", there is no issue with allowing these trivial equilibria to exist.}.
\end{enumerate}

We say that price curves $\f$ \emph{support} an allocation $\x$ if $(\x, \f)$ is a price curve equilibrium. The fundamental question we address in this paper is: what allocations $\x$ can be supported?

\section{Group domination}\label{sec:gdf}
Recall that we require price curves to be continuous and weakly increasing. We wish to theoretically characterize which allocations can be supported by price curves so that we can (1) apply this characterization in our subsequent proofs, and (2) construct an algorithm which can calculate price curves in polynomial time.

The true necessary and sufficient condition for an allocation to be supported by price curves -- and an algorithm to compute them -- is given in Appendix~\ref{sec:laf}. However, this condition (``locked-agent-freeness'') is somewhat unwieldy. Although weakly increasing price curves are sometimes necessary\footnote{Consider an instance with two agents and two goods, each with supply 1. Let the agents' weights be given by $w_{11} = w_{21} = w_{12} = 1$ and $w_{22} = 0$. Nash welfare is maximized by splitting good 1 evenly between the two agents, and allowing agent 1 to purchase an equal quantity of good 2. This only possible if the price of good 2 is zero: otherwise, agent 1 is paying more than agent 2. Recall that the Fisher market equilibrium prices are the dual variables of the convex program for maximizing Nash welfare: thus the price of good 2 being zero corresponds to the fact that the supply constraint for good 2 is not tight in this instance.}, for now we will consider only \emph{strictly} increasing price curves. The corresponding necessary and sufficient condition is the cleaner notion of group-domination-freeness.


\subsection{Group domination}
We have termed the necessary and sufficient condition for the existence of strictly increasing price curves ``group-domination-freeness'' (GDF). To gain intuition for this condition, consider an allocation $\x$ and agents $i,k$. We will say that agent $i$ \emph{dominates} agent $k$ if $\forall j \ x_{ij} \geq x_{kj}$ and there exists $j$ for which this inequality is strict. Observe that this would prevent the existence of strictly increasing price curves supporting allocation $\x$ -- both agents must spend their full budget (otherwise they could buy more of every good, since price curves are continuous), but agent $i$ buys strictly more than agent $k$. A similar scenario arises when considering any two weighted sets of agents $\A,\B$. We can represent these weighted groups as vectors with a non-negative weight\footnote{Note that this is \emph{not} the same weight as the $w_{ij}$ representing an agent's weight for a good.} for each agent, where we require that $\A$ and $\B$ have the same total agent-weight. If for every possible quantity $\tau \in \bbrpos$ of any good $j$, considering only agents purchasing at least $\tau$ of good $j$, the weight of the agents in $\A$ is greater than or equal to the weight of agents in $\B$, then $\B$ can never be made to pay more than $\A$. Essentially, for each additional $\delta$ of any good, as many agents from $\A$ must purchase that $\delta$ as agents from $\B$, so no matter how we price these increments, $\B$ never pays more. If this difference in weights is strict for any ($j,\tau)$ pair, that implies some $\delta$ increment must cost 0 (because the total expenditure of $\A$ and $\B$ must be equal), violating the requirement that price curves be strictly increasing.

Another way to gain intuition for group domination is by analogy to stochastic dominance. Distribution $\A$ stochastically dominates distribution $\B$ if for every possible payoff $\tau \in \bbrpos$, the odds of getting at least $\tau$ from $\A$ are at least as high as the odds of getting at least $\tau$ from $\B$. One consequence of stochastic dominance is that \emph{any} rational agent should prefer $\A$ to $\B$ -- there are no trade-offs, $\A$ is simply better than (dominant over) $\B$. In fact, we can directly consider weighted groups of agents as probability distributions. The total weight of each group must be the same -- without loss of generality, equal to 1. Consider sampling the allocations $x_{ij}$ for any good $j$ with probability equal to the weight of each agent. The probability distribution $\A$ stochastically dominating $\B$ is exactly equivalent to the weighted group $\A$ group-dominating $\B$. Thus not only does group domination create problems for pricing, it can arguably be considered \emph{unfair}, as $\A$ is in some sense \emph{objectively} better-off than $\B$\footnote{See \cite{feller_1968_introduction} for an introduction to stochastic dominance.}.

The formal definition of this condition is below.

\begin{definition}[\textbf{Group-domination-free (GDF)}]\label{def:gd}
Let $\A = (a_1,\ldots, a_n)$ and $\B = (b_1,\ldots, b_n)$ be vectors in $\bbr_{\geq 0}^n$ that assign a (possibly zero) weight to each agent, such that $\sum_{i \in N} a_i = \sum_{i \in N} b_i$. Then $\A$ group-dominates $\B$ in an allocation $\x$ (denoted $\A \succ \B$) if for all $j \in M$ and for any threshold $\tau \in \bbrpos$,
\[
\sum_{i \in N:\ x_{ij} \geq \tau} a_i \geq \sum_{i\in N:\ x_{ij} \ge \tau} b_i
\]
and there exists a ($j,\tau$) pair where the inequality is strict. $\x$ is group-domination-free if there do not exist vectors $\A,\B\in\bbrpos^n$ such that $\A\succ\B$ in $\x$.\footnote{The ``-free/-freeness" suffix may remind some readers of the popular fairness notion envy-freeness; this connection is intended. If one agent does envy another, this constitutes an instance of group domination in the allocation, so GDF implies envy-freeness. However, the reverse is not true: for an agent $i$ to envy agent $k$, $k$ must receive strictly more of every good $i$ cares about; for group domination, the difference need only be strict on one good. All market equilibria are envy-free; GDF is a stronger notion corresponding exactly the the existence of a market equilibrium in this setting.}

\end{definition}

We will also assume without loss of generality that for all $i \in N$, at least one of $a_i$ and $b_i$ is zero, i.e., these are non-overlapping weighted groups: were this not the case, we could define $\aprime$ and $\bprime$ by $a_i' = a_i - \min(a_i, b_i)$ and $b_i' = b_i - \min(a_i, b_i)$, and we would have $\aprime \succ \bprime$ if and only if $\A \succ \B$.

Theorem~\ref{thm:gdf-final} will show that an allocation can be supported with strictly increasing price curves if and only if it is GDF. 

\subsection{Characterization of allocations supported by strictly increasing price curves}\label{sec:characterization-strict}
In order to relate the existence of price curves to GDF, first observe that, for agents with Leontief utilities, the conditions for a price curve equilibrium take on a relatively simple form. Recall that by assumption, the allocation to be considered doesn't violate supply, and each agent purchases goods in exact proportion to her weights $w_{ij}$ (see Section~\ref{sec:fisher-markets}). Then the condition that $x_i \in D_i(\f)$ for all $i$ can be captured by Lemma~\ref{lem:leontief-demand}, whose proof appears in Appendix~\ref{sec:extra-proofs}. Intuitively, agent $i$ fills up her bundle in proportion to her weights until (a) she reaches her budget and (b) there exists a good where buying more would cost more.

\begin{restatable}{lemma}{lemLeontiefDemand}
\label{lem:leontief-demand}
Given price curves $\f$, $x_i \in D_i(\f)$ if and only if both of the following hold:  (a) $C_\f(x_i) = 1$, and (b) there exists $j \in M$ such that for all $\ep > 0$, $f_j(x_{ij} + \ep w_{ij}) > f_j(x_{ij})$.
\end{restatable}

We are now almost ready to prove Theorem~\ref{thm:gdf-final} relating the existence of price curves to GDF. However, the proof is rather intricate, so we begin by giving an intuitive overview thereof. Throughout, we will refer to the example allocation $\x$ shown in Figure~\ref{fig:allocx} to make the argument concrete. (Note that the example allocation shown in the figure implicitly defines a corresponding Leontief utility function for each agent, up to scaling by a constant, since we assume each agent fills up her bundle in exact proportion to her weights $w_{ij}$.)

We will now use this example to illustrate three key observations regarding the existence of strictly increasing price curves supporting an allocation $\x$: (1) Only the points on the price curves corresponding to agent allocations matter. (2) Only the order of the agents along the price curve for each good, not their absolute allocations, matters. (3) The order of the agents can be captured in an \emph{agent-order matrix} such that weighted column and row sums represent agent costs and group dominations, respectively.

First we address observation (1). Consider the possible price curves shown in Figure~\ref{fig:graphx}. Given the price that each agent pays for each good, these are the only points that matter, in the sense that (a) each agent's total cost, which must equal 1, depends only on these points, and (b) an agent must be able to purchase more of a good if the next fixed point on that curve has the same price, and otherwise need not be able to do so, for instance if we make the price curves piece-wise linear as shown. Thus when considering whether price curves are possible, we need only consider the set of prices corresponding to agent allocations.

A similar argument addresses observation (2). As long as we fix the order of points along a price curve, we can change the allocations arbitrarily (assuming they still obey the supply and proportional-purchase assumptions) without changing the prices. Obviously, every agent will still incur a cost of 1, and it will not change whether an agent can purchase more of a good (whether the next point along the curve has the same price).

\begin{figure}
\centering
\begin{subfigure}[b]{.39\textwidth}
\centering
\begin{tabular}{c|ccc}
& good 1 & good 2 & good 3 \\
\hline
agent 1 & 0.6 & 0 & 0.2 \\
agent 2 & 0.3 & 0 & 0 \\
agent 3 & 0.1 & 1 & 0.4 \\
agent 4 & 0 & 0 & 0.4 \\
\end{tabular}
\caption{example allocation $\x$}\label{fig:allocx}
\end{subfigure}
\hfill
\begin{subfigure}[b]{.6\textwidth}
\centering
{\setlength{\tabcolsep}{7pt}
\begin{tabular}{ccccc:c:cccc:ccc}
& \multicolumn{3}{c}{good 1} & \multicolumn{1}{c}{} & \multicolumn{1}{c}{good 2} & \multicolumn{1}{c}{} & \multicolumn{2}{c}{good 3} & \multicolumn{1}{c}{} && \\ [1ex]
\ldelim[{4}{15pt}\hspace{-7pt} & 1 & 1 & 1 && 0 && 1 & 0 && -1 & \hspace{-9pt}\rdelim]{4}{15pt} \\
& 1 & 1 & 0 && 0 && 0 & 0 && -1 & \\
& 1 & 0 & 0 && 1 && 1 & 1 && -1 & \\
& 0 & 0 & 0 && 0 && 1 & 1 && -1 & \\
\end{tabular}
}
\caption{$\x$ represented as an agent-order matrix $A$}\label{fig:matrixx}
\end{subfigure}
\vspace{10pt}
\begin{subfigure}[c]{\textwidth}
\centering
\begin{tikzpicture}[y=.125cm, x=.39cm]
\begin{scope}
	\draw (0,0) -- coordinate (x axis mid) (10,0);
    \draw (0,0) -- coordinate (y axis mid) (0,25);
    \foreach \x in {0.2,0.4,0.6,0.8,1}
    	\draw (\x*10,2pt) -- (\x*10,-2pt)
        node[anchor=north] {\x};
 
	\node[below=0.5cm] at (x axis mid) {allocation of good 1};
	\node[rotate=90,above=0.1cm] at (y axis mid) {price};

	\filldraw (0,0) circle (0.06cm) node[above right]{\hspace{0pt} agent 4};
	\filldraw (1,7) circle (0.06cm) node[right]{agent 3};
	\filldraw (3,15) circle (0.06cm) node[below right]{agent 2};
	\filldraw (6,20) circle (0.06cm) node[below right]{agent 1};
    \draw[dashed] (0,0) -- (1,7) -- (3,15) -- (6,20) -- (6.33,25);
\end{scope}
\begin{scope}[shift={(11.75,0)}]
	\draw (0,0) -- coordinate (x axis mid) (10,0);
    \draw (0,0) -- coordinate (y axis mid) (0,25);
    \foreach \x in {0.2,0.4,0.6,0.8,1}
    	\draw (\x*10,2pt) -- (\x*10,-2pt)
        node[anchor=north] {\x};
 
	\node[below=0.5cm] at (x axis mid) {allocation of good 2};

	\filldraw (0,0) circle (0.06cm) node[above right]{\hspace{17pt} agents 1,2,4};
	\filldraw (10,18) circle (0.06cm) node[above left]{agent 3};
    \draw[dashed] (0,0) -- (10,18);
\end{scope}
\begin{scope}[shift={(23.5,0)}]
	\draw (0,0) -- coordinate (x axis mid) (10,0);
    \draw (0,0) -- coordinate (y axis mid) (0,25);
    \foreach \x in {0.2,0.4,0.6,0.8,1}
    	\draw (\x*10,2pt) -- (\x*10,-2pt)
        node[anchor=north] {\x};
 
	\node[below=0.5cm] at (x axis mid) {allocation of good 3};

	\filldraw (0,0) circle (0.06cm) node[above right]{\hspace{12pt} agent 2};
	\filldraw (2,6) circle (0.06cm) node[right]{\hspace{6pt} agent 1};
    \filldraw (4,15) circle (0.06cm) node[right]{agents 3,4};
    \draw[dashed] (0,0) -- (2,6) -- (4,15) -- (5,26);
\end{scope}
\end{tikzpicture}
\caption{example price curves for allocation $\x$}\label{fig:graphx}
\end{subfigure}
\caption{An illustrative example allocation and the construction of the corresponding agent-order matrix.}\label{fig:gdf_intuition}
\end{figure}
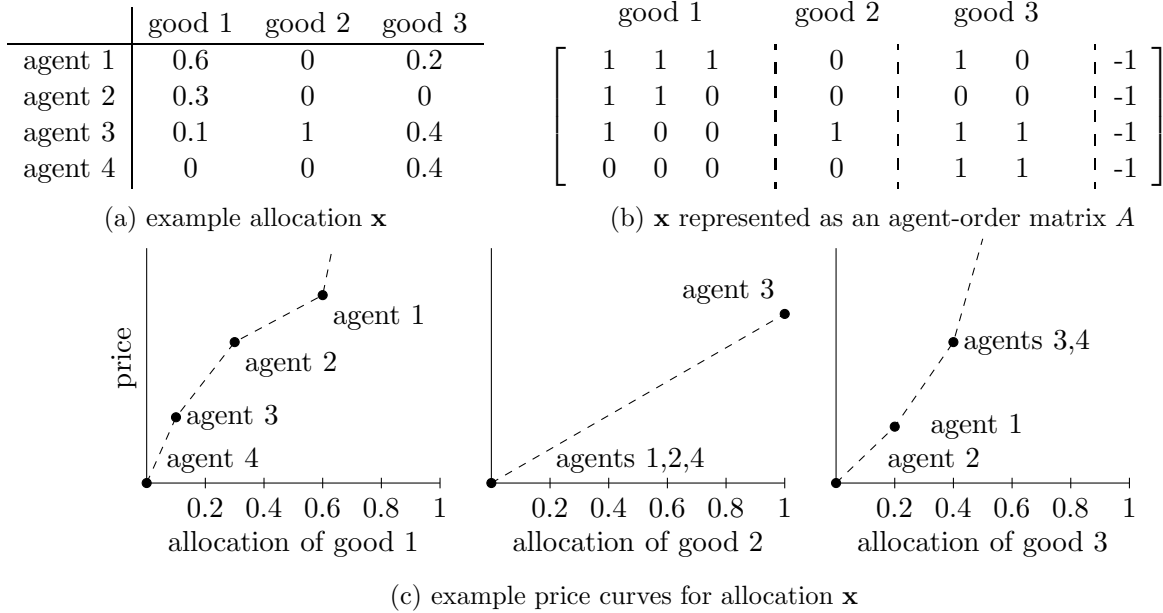

Finally, we come to the more complicated observation (3). We will first lay out how the agent-order matrix is constructed, then illustrate its connection to both prices and group domination. The matrix will have $n$ rows, one for each agent, and a sub-block for each good, as shown in Figure~\ref{fig:matrixx}. Within a sub-block, each column will correspond to a non-zero agent allocation (i.e., the non-zero points shown in Figure~\ref{fig:graphx}). The entry corresponding to agent $i$, good $j$, and allocation threshold $\tau \in \bbrpos$ will equal 1 if $x_{ij} \geq \tau$ and 0 otherwise. Essentially, this will indicate which agent pays the cost of the first, second, etc.\ section of each price curve. Additionally, we append a column of $-1$'s to the end of the matrix. To see the connection to prices, consider a vector $\y$ such that $A \y = \mathbf{0}, \y \neq \mathbf{0}$. For instance, Figure~\ref{fig:row_sum} exhibits such a vector $\y$ for the matrix $A$ shown in Figure~\ref{fig:matrixx}.
$\y$ will represent prices, so we require all the entries to be non-negative, denoted $\y \geq 0$; for strictly increasing price curves, we require $\y$ to be strongly positive\footnote{Recall that a strongly positive vector has every entry greater than 0.}, denoted $\y \gg 0$. Specifically, within each block (corresponding to a good $j$), the first entry represents the cost of increasing from 0 of good $j$ to the first non-zero point on the price curve, the second entry represents the cost of increasing from the first point to the second point, and so on. The last entry in $\y$, which we can assume equals 1 without loss of generality, represents the total cost expended by each agent. Thus $A \y = \mathbf{0}$ ensures that each agent spends exactly 1 unit of money. (Ensuring that condition (b) of Lemma~\ref{lem:leontief-demand} be met is slightly more complicated. However, for strictly increasing price curves, it is trivially satisfied.)

\begin{figure}
\centering
\begin{subfigure}[t]{.42\textwidth}
\centering
$\left[\begin{array}{ccc:c:cc:c}
1&1&1&0&1&0&-1 \\
1&1&0&0&0&0&-1 \\
1&0&0&1&1&1&-1 \\
0&0&0&0&1&1&-1 \\
\end{array}\right] \times
\setlength{\dashlinegap}{2pt}
\left[\begin{array}{c}
0 \\ 1 \\ 0 \\ \hdashline
0 \\ \hdashline
0 \\ 1 \\ \hdashline
1 \\
\end{array}\right] = \mathbf{0}
\setlength{\dashlinegap}{4pt}$
\caption{column sum of $A$ implies price curves \\(in this example, weakly increasing)}\label{fig:row_sum}
\end{subfigure}
\hfill
\begin{subfigure}[t]{.47\textwidth}
\centering
$\left[\begin{array}{cccc}
\phantom{-}1&\phantom{-}1&\phantom{-}1&\phantom{-}0\phantom{-} \\
\phantom{-}1&\phantom{-}1&\phantom{-}0&\phantom{-}0\phantom{-} \\
\phantom{-}1&\phantom{-}0&\phantom{-}0&\phantom{-}0\phantom{-} \\ \hdashline
\phantom{-}0&\phantom{-}0&\phantom{-}1&\phantom{-}0\phantom{-} \\ \hdashline
\phantom{-}1&\phantom{-}0&\phantom{-}1&\phantom{-}1\phantom{-} \\
\phantom{-}0&\phantom{-}0&\phantom{-}1&\phantom{-}1\phantom{-} \\ \hdashline
-1&-1&-1&-1\phantom{-} \\
\end{array}\right] \times
\left[\hspace{-3pt}\begin{array}{c}
\phantom{-}1 \\ -1 \\ \phantom{-}0 \\ \phantom{-}0 \\
\end{array}\right] = 
\setlength{\dashlinegap}{2pt}
\left[\begin{array}{c}
0 \\ 0 \\ 1 \\ \hdashline
0 \\ \hdashline
1 \\ 0 \\ \hdashline
0 \\
\end{array}\right]
\setlength{\dashlinegap}{4pt}$
\caption{row sum of $A$ implies group domination \\(precludes strictly increasing price curves)}\label{fig:col_sum}
\end{subfigure}
\caption{Example row and column sums of the agent-order matrix from Figure~\ref{fig:matrixx}.}\label{fig:sums}
\end{figure}

Thus we can see that the column sums of the agent-order matrix correspond to agent expenditures, where the weight of each column in the sum corresponds to a section of the price curve. \emph{Row} sums, however, correspond to group domination. To see the connection, consider a vector $\z$ such that $A^T \z$ is strictly positive\footnote{Recall that a strictly positive vector has entries in $\bbrpos$ with at least one entry non-zero.}, denoted $A^T \z > \mathbf{0}$. For instance, Figure~\ref{fig:col_sum} exhibits such a vector $\z$ for the matrix $A$ shown in Figure~\ref{fig:matrixx}.
In a given $\z$, the positive entries correspond to the weighted agents in a dominating group $\A$, while the (absolute value of the) negative entries are the weighted agents in group $\B$. Since the last entry of $A^T \z$ must be nonnegative, the total weight of $\B$ is at least as large as that of $\A$. And since $A^T \z > \mathbf{0}$, all the entries are non-negative and at least one other entry must be positive. This means that at every point on a price curve (any $j,\tau$), the weight of group $\A$ purchasing at least $\tau$ of good $j$ is at least as much as the weight of group $\B$ purchasing $\tau$, and for some ($j,\tau$) this is strict. Clearly this is equivalent to $\A \succ \B$.

Having constructed the agent-order matrix and related its column and row sums to prices and group domination, respectively, the final step applies a previously-known duality result equivalent to Farkas' Lemma \cite{perng_2017}, which establishes that valid prices (column sums) exist if and only if group domination (row sums) do \emph{not}. Specifically, we make use of the following result originally due to Stiemke to prove Theorem~\ref{thm:gdf-final}.

\begin{lemma}[1.6.4 in \citep{stoer_1970_convexity}]\label{lem:matrix}
For a commutative, ordered field $\mathbb{F}$, $A$ a matrix over $\mathbb{F}$, the following are equivalent: (1) $A \y = \mathbf{0}, \y \gg \mathbf{0}$ has no solution. (2) $A^{T} \z > \mathbf{0}$ has a solution.
\end{lemma}

\begin{theorem}\label{thm:gdf-final}
Let $\x$ be any allocation that obeys the supply constraints and gives at least one agent a nonempty bundle. Then $\x$ be can supported by strictly increasing price curves if and only if $\x$ is GDF.
\end{theorem}
\begin{proof}
Recall that an allocation $\x$ can be supported if there exist price curves $\f$ such that $x_i \in D_i(\f) \ \forall i \in N$, and $\sum_{i \in N} x_{ij} \leq 1 \ \forall j \in M$ (i.e., $\x$ obeys the supply constraints). The latter condition is satisfied by assumption, and by Lemma~\ref{lem:leontief-demand}, for Leontief utilities and strictly increasing price curves, the former condition holds if and only if the cost $C_\f(x_i) = 1 \ \forall i$.

Let $X_j = \{x_{ij} \mid i \in N\} \setminus \{0\}$ be the set of distinct, non-zero amounts of good $j$ allocated to some agent under $\x$. Label the elements of $X_j$ as $\tau_j^1, \tau_j^2,\ldots, \tau_j^{|X_j|}$ such that $\tau_j^1 < \tau_j^2 < \cdots  < \tau_j^{|X_j|}$. Since $f_j(0) = 0$, $f_j(x \notin X_j)$ in some sense doesn't matter -- we only require that these ``in-between'' areas of the price curve don't violate continuity and are strictly increasing. Thus there exist strictly increasing price curves $\f$ supporting $\x$ if and only if there exist functions $f_j': X_j \to \bbrpos$ such that $0 < f_j'(\tau_j^1) < f_j'(\tau_j^2) < \ldots  < f_j'(\tau_j^{|X_j|}) \ \forall j$ and $C_\f(x_i) = \sum_{j} f_j'(x_{ij}) = 1 \ \forall i$.

Now we are ready to set up the agent-order matrix $A \in \mathbb{Q}^{n \times \left(\sum_j |X_j| + 1\right)}$ to which we will apply Lemma~\ref{lem:matrix}. Since each column will represent an allocation point for a specific good (corresponding to its sub-block), we will write the column indices as $\sum_{\ell < j} |X_\ell| + q$, where $j$ indicates the sub-block and $1 \leq q \leq |X_j|$ is the index within that sub-block.
$$A\left[i, \ \sum_{\ell < j} |X_\ell| + q\right] = 
\begin{cases}
-1 &\text{if } j=m+1, q=1 \text{ (last column)} \\
\phantom{-}0 &\text{if } x_{ij} < \tau_j^{q} \\
\phantom{-}1 &\text{otherwise}
\end{cases}$$
Thus each row of $A$ represents an agent, and each column (except the last) represents one point of the functions $\f'$. Since $\x$ gives at least one agent a nonempty bundle by assumption, $A$ has at least two columns (one allocation point and the column of $-1$'s). We know by Lemma~\ref{lem:matrix} that $\exists \y \gg \mathbf{0}$ such that $A \y = 0$ if and only if there does not exist a $\z$ such that $A^T \z > \mathbf{0}$. To complete the proof, we will show that the former condition is equivalent to the existence of strictly increasing price curves supporting $\x$, and the latter is equivalent to a group domination.

If $\exists \y \gg \mathbf{0}$ such that $A \y = \mathbf{0}$, we may assume without loss of generality that the last entry in $\y$ is 1. Furthermore, define $f_j'(\tau_j^q) - f_j'(\tau_j^{q-1}) = y_{\sum_{\ell < j} |X_\ell| + q}$ (for convenience, define $f_j'(\tau_j^0) = f'_j(0) = 0$). Clearly $\y \gg \mathbf{0}$ is equivalent to the requirement that $0 < f_j'(\tau_j^1) < \ldots  < f_j'(\tau_j^{|X_j|}) \ \forall j$. Additionally,
\begin{align*}
C_\f(x_i) &= \sum_j f_j(x_{ij}) = \sum_j f_j'(x_{ij}) = \sum_j \sum_{q : x_{ij} \geq \tau_j^q} y_{\sum_{\ell < j} |X_\ell| + q} = A_i \y + 1
\end{align*}
Thus $A\y = \mathbf{0}$ is equivalent to the requirement that $C_\f(x_i) = 1 \ \forall i$.

Finally, consider $\z$ such that $A^T \z > \mathbf{0}$. This is equivalent to a group domination $\A \succ \B$, where $a_i = z_i$ if $z_i > 0$, $b_i = -z_i$ if $z_i < 0$, and all other entries are 0. Consider the product of the last column of $A$ with $\z$, which equals $-\sum_i z_i \geq 0$. Without loss of generality, we can assume $\sum_i z_i = 0$, and thus $\sum_i a_i = \sum_i b_i$. If this is not true, then $\B$ would have greater weight than $\A$, and decreasing any weight in $\B$ can only increase coordinates of $A^T \z$ or equivalently widen the gap between $\A$ and $\B$ in terms of group domination. Now observe that for any good $j$ and $\tau \in \bbrpos$, $\sum_{i \in N : x_{ij} \geq \tau} (a_i - b_i)$ is equal to the dot product of column $\sum_{\ell < j} |X_\ell| + q$ of $A$ by $\z$, where $q$ is the largest value such that $\tau_j^q \leq \tau$. This holds because $A\left[i,\sum_{\ell < j} |X_\ell| + q\right]$ is an indicator variable for $x_{ij} \geq \tau_j^q$, and by construction no agent can have an allocation amount between $\tau_j^q$ and $\tau_j^{q+1}$. Therefore $A^T \z > \mathbf{0}$ is equivalent to the requirement that $\sum_{i \in N : x_{ij} \geq \tau} (a_i - b_i) \geq 0$ for all ($j, \tau$) and that for some ($j,\tau$) this inequality is strict, i.e., $A^T \z > \mathbf{0}$ is equivalent to $\A \succ \B$.
\end{proof}

\begin{remark}
Since the matrix $A$ constructed in the proof of Theorem~\ref{thm:gdf-final} is over the rationals, we can also assume that the $\y$ or $\z$ obtained via Lemma~\ref{lem:matrix} are over the rationals. In particular, we can scale $\z$ to obtain $\z' \in \mathbb{Z}^n$ with $A^T z' > 0$. Equivalently, this means that if $\A \succ \B$, we can assume without loss of generality that $a_i, b_i \in \mathbb{Z}$.
\end{remark}


This characterization, in addition to allowing us to prove some of our subsequent results, implies that we can compute price curves (or show that they do not exist) for a particular instance in polynomial time. This is exhibited by the following linear program.

\begin{theorem}\label{thm:gdf-lp}
Given a set of agents $N$, goods $M$, and an allocation $\x \in \bbr_{\ge 0}^{n \times m}$, let $A$ be the corresponding agent-order matrix. In the following linear program, the optimal objective value is strictly positive if and only if there exist strictly increasing price curves supporting $\x$, in which case $\y$ defines such curves.
\begin{align*}
\max_{\y, \eta} \ \eta \quad \quad &\ \\
s.t. \quad A\y =&\ \mathbf{0}\\
	y_k \ge&\ \eta\quad \forall k\\
	y_{-1} =&\ 1
\end{align*}
\end{theorem}
\begin{proof}
As per the proof of Theorem~\ref{thm:gdf-final}, there exist strictly increasing price curves supporting $\x$ if and only if there is a solution to the system $A \y = \mathbf{0}, \y \gg 0$. To turn this into a valid linear program, instead of the strict inequality $y_k > 0$ for each coordinate of $\y$, we write $y_k \geq \eta$ and attempt to maximize $\eta$. Furthermore, we restrict the final entry of $\y$ as $\y_{-1} = 1$, since otherwise $\y$ can be scaled arbitrarily. If there is a solution with $\eta > 0$, then $\y$ corresponds to price curves as before, with each entry representing the difference in price between adjacent allocation amounts. These points simply need to be connected, e.g., in a piecewise linear fashion, to constitute valid price curves.
\end{proof}

One may wonder if Theorem~\ref{thm:gdf-final} generalizes to other classes of utility functions. Unfortunately, the answer in general is no. Example~\ref{counter:ef-gdf} gives an instance with linear utilities that is GDF, but cannot be supported by price curves.

In Section~\ref{sec:maxmin}, we will show how the group-domination-freeness concept can be useful for proving that allocations of interest can be supported by price curves: specifically, allocations with optimal (or near optimal) max-min welfare. But first, a word about unequal budgets.

\subsection{Unequal budgets}

It turns out that the characterization theorem of the previous section easily generalizes to agents with unequal budgets. Since price curves are strictly increasing, the only additional requirement for an allocation $\x$ to be supported is that each agent spends her entire budget $B_i$. In the agent-order matrix, the last column of $-1$'s corresponded to each agent's expenditure, so we simply need to replace $-1$ with $-B_i$ for each row $i$.

Following Lemma~\ref{lem:matrix} with the modified agent-order matrix, the if-and-only-if characterization becomes ``budget-weighted group-domination-freeness''. A budget-weighted group domination still requires that for all ($j,\tau$),
\[\sum_{i \in N : x_{ij} \geq \tau} a_i \geq \sum_{i \in N : x_{ij} \geq \tau} b_i\]
and that there exists $j,\tau$ where the inequality is strict. The only difference is that instead of requiring both groups to have the same total weight, that weight is now scaled by each agent's budget. That is, $\sum_{i \in N} a_i B_i = \sum_{i \in N} b_i B_i$. Note that when $B_i = 1$ for all $i$, this recovers the definition of group domination.

\section{CES welfare}\label{sec:ces-duality}

In this section, we consider CES welfare functions:
\[
\Phi_{CES}(\x) = \Big(\sum_{i \in N} u_i(x_i)^\rho \Big)^{1/\rho}
\]
This section contains our second main result: that in the bandwidth setting (i.e., agents have Leontief utilities where $w_{ij} \in \{0,1\}$ for all $i \in N$, $j \in M$), for any $\rho \in (-\infty, 0)\cup (0,1)$, any maximum CES welfare allocation can be supported by price curves (Theorem~\ref{thm:bandwidth}). We present this result this in Section~\ref{sec:ces-main}. Next, we discuss why we are optimistic about the possibility of a simple decentralized primal-dual algorithm for computing these price curves, similar to the work of Kelly et al.~\cite{kelly_1998_rate} (Section~\ref{sec:decentralized}). We also give a converse of sorts to Theorem~\ref{thm:bandwidth} (Section~\ref{sec:ces-converse}), and briefly discuss the case of unequal budgets (Section~\ref{sec:ces-unequal-budgets}). Throughout this section, we let $R_i = \{j \in M: w_{ij} = 1\}$ for brevity.

\subsection{Main CES welfare result}\label{sec:ces-main}

We now state and prove Theorem~\ref{thm:bandwidth}. Our proof uses the dual of the convex program for maximizing CES welfare to construct explicit price curves that support a maximum CES welfare allocation. The price curves take the very simple form of $f_j(x) = q_j x^{1-\rho}$ for constants $q_1,\dots,q_m$ that are derived from the dual. 

\begin{restatable}{theorem}{thmBandwidth}\label{thm:bandwidth}
If $w_{ij} \in \{0,1\}$ for all $i \in N$ and $j \in M$, then for any $\rho \in (-\infty, 0) \cup (0,1)$, any maximum CES welfare allocation can be supported by price curves of the form $f_j(x) = q_j x^{1-\rho}$ for each $j \in M$.
\end{restatable}

\begin{proof}
The proof proceeds in three steps.

\textbf{Step 1: Setting up the convex program.} We begin by writing the following program to maximize CES welfare:
\begin{alignat*}{2}
\max\limits_{\substack{\x \in \bbr_{\ge 0}^{n\times m},\\ \bfu = (u_1...u_n) \in \bbrpos^n}} &\ \Big(\sum_{i \in N} u_i^\rho  \Big)^{1/\rho} \\\nonumber
s.t.\ &\ u_i \le x_{ij} \quad\quad &&\ \forall i \in N, j \in R_i\\
&\ \sum\limits_{i \in N} x_{ij}\leq s_j\quad\quad &&\ \forall j \in M 
\end{alignat*}
We could also have written the first constraint as $u_i w_{ij} \le x_{ij}$, but since $w_{ij} \in \{0,1\}$, the above formulation is equivalent. Also, the objective $\Big(\sum_{i \in N} u_i^\rho  \Big)^{1/\rho}$ is concave for any $\rho \in (-\infty, 0)\cup(0,1)$, so the resulting program is convex.

We can remove the exponent of $1/\rho$ from the objective without affecting the optimal point: the optimal value may be affected, but the optimal solution (i.e., the $\argmax$) will not. When $\rho$ is negative, this changes the program to a minimization program, but this can be handled by adding a factor of $1/\rho$ to the objective.\footnote{We add a factor of $1/\rho$ instead of $\rho$ because this will slightly simplify the analysis.} Thus consider a new convex program with objective function $\max\limits_{\substack{\x \in \bbr_{\ge 0}^{n\times m}, \bfu \in \bbrpos^n}} \frac{1}{\rho}\sum_{i \in N} u_i^\rho$, and the same constraints.

Next, we write the Lagrangian of the new program. Let $\lam_{ij}$ be the Lagrange multiplier associated with the constraint $u_i \le x_{ij}$ and let $q_j$ be the Lagrange multiplier associated with the constraint $\sum_{i \in N} x_{ij}\leq s_j$. We will use $\bflam$ and $\q$ to denote the vectors of all such Lagrange multipliers. Then the Lagrangian is given by
\[
L(\x,\bfu, \bflam, \q) = \frac{1}{\rho}\sum_{i \in N} u_i^\rho - \sum_{i \in N} \sum_{j \in R_i} \lam_{ij} (u_i - x_{ij}) - \sum_{j \in M} q_j\Big(\sum_{i \in N} x_{ij} - s_j\Big)
\]
Consider any maximum CES welfare allocation: this corresponds to a point $(\mathbf{x^*}, \mathbf{u^*})$ which is optimal for the primal. We have strong duality by Slater's condition, so there must exist $\boldsymbol{\lam^*}$ and $\mathbf{q^*}$ such that $(\mathbf{x^*}, \mathbf{u^*}, \boldsymbol{\lam^*}, \mathbf{q^*})$ is optimal for $L$.

\textbf{Step 2: Using the KKT conditions.} The KKT conditions imply that the gradient of $L$ evaluated at $(\mathbf{x^*}, \mathbf{u^*}, \boldsymbol{\lam^*}, \mathbf{q^*})$ must be zero for every variable with a positive value. Specifically, for each variable $y$, either  $\frac{\partial L}{\partial y} = 0$, or $y = 0$ and $\frac{\partial L}{\partial y} \le 0$.

First, we have $\frac{\partial L}{\partial u_i} (\mathbf{x^*}, \mathbf{u^*}, \boldsymbol{\lam^*}, \mathbf{q}^*) = {u_i^*}^{\rho - 1} - \sum_{j \in R_i} \lam_{ij}^* = 0$ for every $i \in N$ with $u_i^* > 0$. Suppose $u_i^* = 0$ for some $i \in N$: since $\rho - 1 < 0$, we would have ${u_i^*}^{\rho - 1} = \infty$, which contradicts $\frac{\partial L}{\partial y} \le 0$. Thus $u_i^* > 0$, so $u_i^* = (\sum_{j \in R_i} \lam_{ij}^*)^{\frac{1}{\rho -1}}$  for all $i \in N$. 

Similarly, $\frac{\partial L}{\partial x_{ij}} (\mathbf{x^*}, \mathbf{u^*}, \boldsymbol{\lam^*}, \mathbf{q}^*)= \lam_{ij}^*- q_j^* = 0$ for every $i \in N$ and $j \in R_i$ with $x_{ij} > 0$. Since $u_i^* > 0$ for all $i \in N$, we must have $x_{ij}^* > 0$ for all $j \in R_i$. Therefore $\lam_{ij}^* = q_j^*$ for all $i \in N, j \in R_i$, so $u_i^* = \big( \sum_{j \in R_i} q_j^* \big)^{\frac{1}{\rho -1}}$. It will also be helpful to consider $\frac{\partial L}{\partial x_{ij}} (\mathbf{x^*}, \mathbf{u^*}, \boldsymbol{\lam^*}, \mathbf{q}^*)$ for $j \not \in R_i$: in this case, we have $\frac{\partial L}{\partial x_{ij}} (\mathbf{x^*}, \mathbf{u^*}, \boldsymbol{\lam^*}, \mathbf{q}^*) = q_j = 0$ whenever $x_{ij}^* > 0$.

Next, the KKT conditions also imply that $(\mathbf{x^*}, \mathbf{u^*}, \boldsymbol{\lam^*}, \mathbf{q^*})$ satisfy \emph{complementary slackness}, meaning that the Lagrange multiplier of any non-tight constraint is equal to 0. We are specifically interested in the constraint $x_{ij}^* \le u_i^*$ for $j \in R_i$: either $\lambda_{ij}^* = q_j^* = 0$, or
\[
x_{ij}^* = u_i^* = \Big( \sum_{j \in R_i} q_j^* \Big)^{\frac{1}{\rho -1}}
\]

\textbf{Step 3: Constructing the price curves.} We now use the Lagrange multipliers $\mathbf{q^*}$ to construct explicit price curves. We define $f_j(x)$ by $f_j(x) = q_j^* x^{1-\rho}$. Since $\rho \in (-\infty, 0)\cup(0,1)$, we have $1-\rho > 0$, so these price curves are in fact increasing. We claim that $(\mathbf{x^*}, \f)$ is a price curve equilibrium. To see this, we explicit compute the cost of agent $i$'s bundle $x_i^*$:
\[
C_\f(x_i^*) =\sum_{\ell \in M} q_\ell^* {(x_{i\ell}^*)}^{1-\rho} =\sum_{\ell: q_\ell, x_{i\ell}\ne 0} q_\ell^* {(x_{i\ell}^*)}^{1-\rho} =
 \sum_{\ell: q_\ell, x_{i\ell}\ne 0} q_\ell^* \Big( \sum_{j \in R_i} q_j^* \Big)^{\frac{1-\rho}{\rho -1}} = \sum_{\ell: q_\ell, x_{i\ell}\ne 0} \frac{q_\ell^*}{\sum\limits_{j \in R_i} q_j^*}
\]
To show that $C_\f(x_i^*) = 1$, we just need to show that $\sum_{j: q_j, x_{ij}\ne 0} q_j^* = \sum_{j \in R_i} q_j^*$. Clearly $\sum_{j: q_j, x_{ij}\ne 0} q_j^* = \sum_{j: x_{ij}\ne 0} q_j^*$. Since $u_i^* > 0$, we have $x_{ij} \ne 0$ for each $j \in R_i$, so $\sum_{j: x_{ij}\ne 0} q_j^* \ge \sum_{j \in R_i} q_j^*$. To show that the reverse inequality holds, it suffices to show that whenever $j \not\in R_i$ and $x_{ij} \ne 0$, $q_j = 0$. This is exactly one of the things we showed via the KKT conditions in Step 2.

Thus we have shown that $C_\f(x_i^*) = 1$, so $x_i^*$ is affordable to agent $i$ for all $i \in N$. Furthermore, since $u_i^* = \Big( \sum_{j \in R_i} q_j^* \Big)^{\frac{1}{\rho -1}}$ is finite, there must exist $j \in R_i$ with $q_j^* > 0$. Thus there is at least one good $j \in R_i$ such that buying more would cost more money, so by Lemma~\ref{lem:leontief-demand}, $x_i^*$ is in agent $i$'s demand set. We also know that $\sum_{j \in M} x_{ij}^* \leq 1$, since $\mathbf{x^*}$ is a feasible solution to the primal. Therefore $(\mathbf{x^*}, \f)$ is a price curve equilibrium.
\end{proof}

The structure of the price curve themselves ($f_j(x) = q_j^* x^{1-\rho}$) is also interesting when we consider the interpretation of the parameter $\rho$: the smaller $\rho$ is, the more we care about agents with small utility. Recall that taking of $\rho \to -\infty$ yields max-min welfare, where we only care about the minimum utility. When $\rho = 1$, we have utilitarian welfare, where we only care about overall efficiency. This roughly corresponds to caring more about agents with higher utility. The limit as $\rho \to 0$ corresponds to Nash welfare, which is a mix of caring about both agents with low utility and those with high utility.

We know that maximum Nash welfare allocations are supported by linear price curves, i.e., those with constant marginal prices. When $\rho < 0$, these marginal prices are increasing, making it easier for agents who are buying less of each good. Since $w_{ij} \in \{0,1\}$, $u_i(x_i) = x_{ij}$ whenever $w_{ij} \ne 0$, so the agents who are buying less are also the ones with lower utility. Thus price curves of this form for $\rho < 0$ are benefiting the agents with low utility. Furthermore, the smaller $\rho$ is, the faster marginal prices grow, which corresponds to favoring agents with low utility even more. On the other hand, when $\rho > 0$, these marginal prices are decreasing. This favors agents with higher utility, which is consistent with the interpretation of the CES welfare function with $\rho > 0$.

\subsection{Decentralized primal-dual updates}\label{sec:decentralized}

In this section, we discuss how the simple structure of these price curves suggests a natural decentralized primal-dual algorithm for computing said price curves. As established by Theorem~\ref{thm:bandwidth}, the price of buying $x$ of good $j$ will be $q_j x^{1-\rho}$. Specifically, for each $j \in M$, $q_j$ be the Lagrange multiplier for the convex program for maximizing CES welfare with respect to that specific value of $\rho$\footnote{In general, different values of $\rho$ will lead to different optimal allocations and Lagrange multipliers.}. 

For bandwidth allocation, each agent's demand set depends only on the prices curves for goods $j \in R_i$, i.e., the goods she cares about. For price curves of the form $f_j(x) = q_j x^{1-\rho}$, this means that each agent's demand given price curves $\f$ depends only on the dual prices $q_j$ for goods $j \in R_i$. This means that given price curves $\f$, agents can update their demands in a decentralized fashion. Furthermore, the price set by each link should depend only on the flow through that link, i.e., $\{x_{ij}: i \in N, j \in R_i\}$. This means that given agent demands, each link can update its dual price $q_j$ in a decentralized way (typically by raising the price if the demand is less than the supply, and increasing the price if the demand exceeds the supply). This suggests a simple decentralized primal-dual algorithm, where on each step, each agent updates her primal allocation $x_i$ in response to the dual prices $q_j$ for $j \in R_i$, and each link updates its dual price in response to primal allocations $x_i$. This is similar to the work of Kelly et al.~\cite{kelly_1998_rate}.

This type of algorithm is also called a \emph{\tat}. One recent approach to \tat makes use of the fact that the equilibrium prices are the Lagrange multipliers in the convex program to maximize Nash welfare, and gives a \tat process that is akin to gradient descent on the dual program~\cite{cheung_tatonnement_2013}. This approach also seems promising for our setting, since $q_1\dots q_m$ are exactly the Lagrange multipliers in the convex program for maximizing CES welfare. We leave this as an open question.

\subsection{A converse to Theorem~\ref{thm:bandwidth}}\label{sec:ces-converse}

In this section, we give a converse of sorts to Theorem~\ref{thm:bandwidth}: if an allocation $\x$ can be supported by price curves $\f$ of the form $f_j(x) = q_j^* x^{1-\rho}$, and the supply is exhausted for any good with nonzero price, then $\x$ must be a maximum CES welfare allocation. The requirement that the supply be exhausted for any good with nonzero price (i.e.,  $\sum_{i \in N} x_{ij} = s_j$ whenever $q_j \ne 0$) is analogous to the second condition in definition of Fisher market (i.e., standard linear pricing) equilibrium given in Section~\ref{sec:fisher-markets}.

The proof of Theorem~\ref{thm:ces-converse} essentially hinges on the fact that when strong duality holds for a convex program, the KKT conditions are sufficient for optimality. This is analogous to the proof of Theorem~\ref{thm:bandwidth}, which is based on the fact that the KKT conditions are necessary for optimality. The formal proof appears in Appendix~\ref{sec:extra-proofs}.

\begin{restatable}{thm}{bandwidthConverse}\label{thm:ces-converse}
Suppose $(\xs, \f)$ is a price curve equilibrium where for all $j \in M$, $f_j(x) = q_j^* x^{1-\rho}$ for $\rho \in (-\infty, 1)$ and nonnegative constants $q_1^*\dots q_m^*$. If $\sum_{i \in N} x^*_{ij} = s_j$ whenever $q_j \ne 0$, then $\xs$ is a maximum CES welfare allocation.
\end{restatable}

\subsection{Unequal budgets}\label{sec:ces-unequal-budgets}

Finally, we address the setting where agents may have different amounts of money to spend. Let $B_i$ be agent $i$'s budget. If we instead consider the \emph{budget-weighted CES welfare} $\Phi_{CES}(\x) = \big(\sum_{i \in N} B_i u_i(x_i)^\rho \big)^{1/\rho}$, then the proof of Theorem~\ref{thm:bandwidth} extends directly. Duality tells us that agent $i$'s utility must be $u_i(x_i) = \Big(\frac{1}{B_i} \sum_{j\in R_i} q_j^* \Big)^{\frac{1}{\rho - 1}}$. By using the same price curve form of $f_j(x) = q_j^* x^{1-\rho}$, we get $C_\f(x_i) = \sum_{\ell \in R_i} \mfrac{B_i q_\ell^*}{\sum\limits_{j \in R_i}q_j^*} = B_i$, so agent $i$ is indeed spending exactly her budget. This can be used to show that any allocation with maximum budget-weighted CES welfare can be supported by price curves.

A social planner may prefer to give the same weight to each agent's utility, even if the budgets are not the same. Unfortunately, allocations with optimal unweighted CES welfare cannot be supported (at least not exactly) when agents have different budgets. To see this, consider two agents with different budgets and a single good: whichever agent has more money must receive a larger portion of the good. But assuming the agents have the same weight for that good (which holds in the bandwidth allocation setting or when weights are normalized somehow), the unweighted CES welfare optimum would give each agent the same amount. This is analogous to the Fisher market setting: the Fisher market equilibria for unequal budgets are exactly the allocations which maximize the budget-weighted Nash welfare.

\section{Conclusion}\label{sec:conclusion}

In this paper, we analyzed price curves in several different settings, focusing on agents with Leontief utilities. Our first main result was that for strictly increasing price curves, an allocation can be supported if and only if it is GDF. We proved this by defining the agent-order matrix, and using duality theorems to show the existence of a strongly positive solution to a particular system of linear equations. Our second main result was that in the bandwidth allocation setting, the maximum CES welfare allocation can be supported by price curves. These price curves took the simple form of $f_j(x) = q_j x^{1-\rho}$. This is contrast to the standard linear pricing setting, where only maximum Nash welfare allocations can be equilibria.

There are many possible directions for future research. The first is the possibility of a simple primal-dual \tat for price curves, as discussed in Section~\ref{sec:decentralized}. We think that the approach of~\cite{cheung_tatonnement_2013} seems especially promising in this regard.

A second possible direction is studying price curves for other classes of agent utilities, and in particular, linear utilities. We will discuss in Appendix~\ref{sec:counterexamples} some of the challenges that linear utilities pose for analyzing price curves, but perhaps everything would fall into place with the right framework. 

Third, future research could consider \emph{quasilinear} utilities. In this paper, we assumed that agents have fixed budgets, and have no value for leftover money (i.e., ``fake money"). In the quasilinear setting (i.e., ``real money"), agents also choose how much to spend, and each agent's utility function is equal to the value she derives from her bundle minus the amount she pays. It would be interesting to see whether our results extend to that setting.

Last but not least, we are intrigued by the connection between GDF and the agent-order matrix and duality theorems, and we wonder if this connection could be useful for other resource allocation problems as well.


\bibliographystyle{plain}
\bibliography{refs}

\section*{Acknowledgements}

This research was supported in part by NSF grant CCF-1637418, 
ONR grant N00014-15-1-2786, and the NSF Graduate Research Fellowship under grant DGE-1656518.

\appendix

\section{Max-min welfare}\label{sec:maxmin}

In this section, we show that under mild assumptions, price curves can support allocations with either optimal max-min welfare, or arbitrarily close to optimal max-min welfare. As before, we assume that agents have Leontief utility functions. Also, we refer to an allocation with optimal max-min welfare as a max-min allocation. 

The first thing we observe is that when agent weights are unconstrained in magnitude, there is no hope to support any approximation of max-min welfare. Consider a single good and two agents with weights $w_{11}$ and $w_{21}$ on that good. In this case, each agent $i$'s utility is just $x_{i1}/w_{i1}$, so the max-min welfare of an allocation $\x$ is $\min(\frac{x_{11}}{w_{11}}, \frac{x_{21}}{w_{21}})$. Now imagine that $w_{11}$ is much larger than $w_{21}$: agent 1 needs significantly more of the good to achieve the same utility as agent 2. Then any max-min allocation (or even any decent approximation) must give more of the good to agent 1 than agent 2. But since agents have the same budgets, any price curve equilibrium must result in each agent receiving half of the supply of good 1, which is a contradiction.

Thus in order to have any hope of even approximately supporting a max-min allocation, the agent weights must be normalized in some way. Theorem~\ref{thm:maxmin-full} states that under a quite general normalization assumption, we can support a max-min allocation.

\begin{theorem}\label{thm:maxmin-full}
Suppose there exist strictly increasing functions $g_1,\dots,g_m$ such that for all $i \in N$, $\sum_{j\in M} g_j(w_{ij}) = 1$. Then there exists a max-min allocation that can be supported by price curves.
\end{theorem}

\begin{proof}
Since the max-min welfare of an allocation is determined by the minimum agent utility, the max-min welfare cannot be improved by making any agent's utility higher than any other. Similarly, since each agent's utility is determined by $\min_{j \in M} x_{ij}/w_{ij}$, the max-min welfare cannot be improved by allocating goods to an agent outside of her desired proportions. Thus there exists a max-min allocation $\x$ where all agents have the same utility $u$, and where $x_{ij} = u\cdot w_{ij}$ for all $i \in N$ and $j \in M$.

Since GDF is invariant to scaling by constants, this implies that $\x$ is GDF if and only if \emph{the weight vectors themselves} are GDF. That is, $\x$ is GDF if and only if the allocation $\mathbf{x'}$ defined by $x'_{ij} = w_{ij}$ is GDF. One realizes that the assumption of $\sum_{j\in M} g_j(w_{ij}) = 1$ for all $i \in N$ is \emph{literally assuming that there exist (strictly increasing) price curves} that support the allocation $\mathbf{x'}$. Thus $\mathbf{x'}$ is GDF by Theorem~\ref{thm:gdf-final}, so $\x$ is GDF, which completes the proof.
\end{proof}


One natural corollary of Theorem~\ref{thm:maxmin-full} is the following:

\begin{corollary}\label{cor:maxmin-Lq}
Suppose there exists some $q \geq 1$ so that $\sum_{j \in M} w_{ij}^q = 1$ for all $i \in N$. Then there exists a max-min allocation that can be supported by price curves.
\end{corollary}

Theorem~\ref{thm:maxmin-full} has an interesting conceptual implication. We can think of price curves themselves as a sort of ``norm" on the allocation, and any allocation for which there is a ``norm" which assigns the same value to each agent's bundle is reasonable enough that it can be supported by price curves. The previous statement can be rephrased as ``an allocation can be supported by price curves if and only if there exist price curves which assign the same cost to each agent's bundle", and so is functionally a tautology. Since there exists a max-min allocation which is a constant scaling of the agent weights, this near-tautology carries over.


One final observation is that there are some interesting norms, such as the $L_\infty$ norm, which cannot be written as the sum of increasing functions. In fact, there are cases where no max-min allocation can be supported when agent weights have the same $L_\infty$ norm.\footnote{The $L_{\infty}$ norm is defined as $\max_{j \in M} w_{ij}$.} Furthermore, the following counterexample falls under the even simpler bandwidth allocation setting: $w_{ij} \in \{0,1\}$ for all $i,j$.

\begin{theorem}\label{thm:maxmin-infty-counterexample}
There exist instances where $w_{ij} \in \{0,1\}$ for all $i\in N$ and $j\in M$, but no max-min allocation can be supported.
\end{theorem}

\begin{proof}
Consider an instance with three agents and two goods, each with supply 1. Let the agent weights be given by the following table:
\[
\begin{tabular}{ c|ccc} 
& agent 1 & agent 2 & agent 3\\
\hline
 good 1 & 1 & 0 & 1 \\
good 2 & 0 & 1 & 1\\
\end{tabular}
\]
The unique max-min allocation is $x_{11} = x_{22} = x_{31} = x_{32} = \frac{1}{2}$. Thus any price curves $f_1, f_2$ must satisfy $C_\f(x_1) = f_1(\frac{1}{2}) = 1$,  $C_\f(x_2) = f_2(\frac{1}{2}) = 1$. But then $C_\f(x_3) = f_1(\frac{1}{2}) + f_2(\frac{1}{2}) = 2$, which is a contradiction. Thus no max-min allocation can be supported. 
\end{proof}


The good news is that the $L_\infty$ norm can be approximated to arbitrary precision by $L_q$ norms, leading to the following theorem. We use $\Phi_{MM}(\x) = \min_{i \in N} u_i(x_i)$ to denote the max-min welfare of allocation $\x$.

\begin{restatable}{thm}{maxminInftyApx}\label{thm:maxmin-infty-apx}
Suppose that $\max_{j \in M} w_{ij} = 1$ for all $i \in N$. Then for every $\epsilon > 0$, there exists an allocation $\x$ that can be supported by price curves where $\Phi_{MM}(\x) \geq (1-\epsilon) \max_{\mathbf{x'}} \Phi_{MM}(\mathbf{x'})$.
\end{restatable}

\begin{proof}
Let $w'_{ij}$ be rescaled versions of $w_{ij}$ so that they are $L_q$-normed for a $q$ to be chosen later. Specifically, let $\alpha_i = (\sum_{j \in M} w_{ij}^q)^{1/q}$, and let $w'_{ij} = w_{ij}/\alpha_i$.

Note that $\sum_{j \in M} w_{ij}'^q = 1$ for all $i \in N$. By Corollary~\ref{cor:maxmin-Lq}, there exists an allocation with optimal max-min welfare with respect to weights $w'_{ij}$ that can be supported by price curves. Let $\x$ be this allocation. Then for all $j \in M$ and all other allocations $\mathbf{x'}$,

\begin{align*}
\min_{i \in N}\ \frac{x_{ij}}{w'_{ij}} \geq&\ \min_{i \in N}\ \frac{x'_{ij}}{w'_{ij}}\\
\min_{i \in N}\ \frac{\alpha_i x_{ij}}{w_{ij}} \geq&\ \min_{i \in N}\ \frac{\alpha_i x'_{ij}}{w_{ij}}\\
\min_{i \in N} \alpha_i u_i(x_i) \ge&\ \min_{i \in N} \alpha_i u_i(x_i')
\end{align*}

In particular, let $\mathbf{x^*}$ be the allocation maximizing max-min welfare with the respect to the true weights $w_{ij}$: then $\min_{i \in N} \alpha_i u_i(x_i) \ge \min_{i \in N} \alpha_i u_i(x_i^*)$. Since $u_i(x_i^*) \geq \Phi_{MM}(\mathbf{x^*})$ by definition, we have $\min_{i \in N} \alpha_i u_i(x_i) \ge \Phi(\mathbf{x^*}) \min_{i \in N} \alpha_i$.

Therefore for all $k \in N$, $\alpha_k u_k(x_k) \ge \Phi(\mathbf{x^*}) \min_{i \in N} \alpha_i$. Therefore
\[
u_k(x_k) \ge \Phi(\mathbf{x^*}) \frac{\min_{i \in N} \alpha_i}{\alpha_k}
\]
and so
\[
\Phi_{MM}(\x) \ge \Phi(\mathbf{x^*}) \frac{\min_{i \in N} \alpha_i}{\max_{i \in N} \alpha_i}
\]

It remains to show that there exists $q \geq 1$ such that $\frac{\min_{i \in N} \alpha_i}{\max_{i \in N} \alpha_i} \ge 1 - \epsilon$. This follows from the fact that $\lim\limits_{q\to \infty} \alpha_i = (\sum_{j \in M} w_{ij}^q)^{1/q} = 1$ for all $i \in N$, which completes the proof.
\end{proof}


\section{Characterization of allocations supported by weakly increasing price curves}\label{sec:laf}

In Section~\ref{sec:gdf}, we showed that an allocation can be supported with \emph{strictly} increasing price curves if and only it is GDF. In this section, we provide the analogous necessary and sufficient condition for the case where any (continuous, weakly increasing) price curves are permitted. This boils down to what we called \emph{locked-agent-freeness} (LAF). LAF is not a particularly interesting condition on its own -- though as with GDF it implies a polynomial time algorithm for finding price curves -- but it is crucial in allowing us to prove that maximum CES welfare allocations can be supported.


For an allocation $\x$, we wish to determine whether there exist price curves $\f$ such that $(\x, \f)$ is a price curve equilibrium. Assuming $\x$ obeys the supply constraints, we just need to determine whether there exist price curves $\f$ such that $x_i \in D_i(\f)$ for all $i \in N$.

%
%

Recall that $\A\succ\B$ if for all $j \in M$ and $\tau\in\bbrpos$, $\sum_{i\in N: x_{ij} \ge \tau} (a_i - b_i) \ge 0$, and there exists a $(j,\tau)$ pair such that the inequality is strict. As discussed in Section~\ref{sec:gdf}, this implies that the aggregate spending of $\A$ is at least that of $\B$ for any $\f$, i.e.,
\[
\sum_{i \in N} (a_i - b_i) C_\f(x_i) \ge 0
\]
for any price curves $\f$. Furthermore, we argued that for strictly increasing $\f$, the inequality is strict, so $\B$ cannot be made to pay as much as $\A$. When we allow weakly increasing price curves, $\A \succ \B$ simply implies that, for any marginal price where $\A$ would have to pay strictly more than $\B$, that marginal price must be zero.



We still need to ensure that $x_i \in D_i(\f) \ \forall i \in N$, i.e., that every agent spends her full budget and cannot get more utility for free (Lemma~\ref{lem:leontief-demand}). This requirement can be expressed by \emph{locked-agent-freeness}.

\begin{definition}[Locked-agent-free (LAF)]\label{def:agent-locked}
For simplicity, we define two meanings of ``locked'':
\begin{itemize}
\item Agent $i$ is locked in an allocation $\x$ if there exists a domination $\A \succ \B$ such that for all $j \in M$ where $x_{ij} > 0$, and all sufficiently small $\ep > 0$, $\A \succ \B$ is strict at $(j, x_{ij} + \ep)$.
\item The \emph{allocation} is locked if there exists $\A \succ \B$ which is strict at \emph{every} $(j,\tau)$ for $\tau \in (0, \max_{i} x_{ij}]$.
\end{itemize}
If nothing is locked in allocation $\x$, we say that $\x$ is locked-agent-free (LAF).
\end{definition}

Intuitively, an agent being locked implies that the cost to increase her allocation must be zero, which will violate condition (b) of Lemma~\ref{lem:leontief-demand}. The \emph{allocation} being locked implies that all marginal prices must be zero, and thus all price curves must be identically zero. Clearly, any non-LAF allocation cannot be supported by price curves. Perhaps surprisingly, the opposite directly holds as well, as stated by Theorem~\ref{thm:laf}.

The proof of Theorem~\ref{thm:laf} is similar to the proof of Theorem~\ref{thm:gdf-final} for strictly increasing price curves. The main difference is that strictly increasing price curves trivially satisfy condition (b) of Lemma~\ref{lem:leontief-demand}, preventing any agent from getting more utility for free. For weakly increasing price curves, however, we need to add a constraint specifically to ensure that condition is satisfied. Thus in addition to the agent-order matrix, we will define a \emph{marginal-cost matrix} to ensure that no agent has a marginal cost of zero to increase her utility. In order to incorporate this matrix, we use a more general duality result than Lemma~\ref{lem:matrix} (although still equivalent to Farkas's Lemma \cite{perng_2017}), this one due to Motzkin. Recall that $\mathbf{v} > \mathbf{0}$ denotes a strictly positive vector, and $\mathbf{v} \gg \mathbf{0}$ strongly positive.

\begin{lemma}[1.6.1 in \cite{stoer_1970_convexity}]\label{lem:transposition}
For matrices $A,B,C$ over $\mathbb{R}$, the following are equivalent.
\begin{enumerate}
\item $A \y = \mathbf{0}, B \y \geq \mathbf{0}, C \y \gg \mathbf{0}$ has no solution
\item $A^{T} \mathbf{u} + B^T \mathbf{v} + C^T \mathbf{w} = \mathbf{0}, \mathbf{v} \geq \mathbf{0},\mathbf{w} > \mathbf{0}$ has a solution
\end{enumerate}
\end{lemma}

\begin{theorem}\label{thm:laf}
Let $\x$ be an allocation which obeys the supply constraints and gives a nonempty bundle to at least one agent. Then $\x$ can be supported by weakly increasing price curves if and only if it is LAF.
\end{theorem}
\begin{proof}
Recall that an allocation $\x$ is supported by price curves $\f$ if $x_i \in D_i(\f) \ \forall i \in N$, and $\sum_{i \in N} x_{ij} \leq 1 \ \forall j \in M$. The latter condition is satisfied by assumption, and by Lemma~\ref{lem:leontief-demand}, for Leontief utilities, the former condition holds if and only if the cost $C_\f(x_i) = 1$ and there exists $j \in M$ such that $\forall \ep > 0 \ f_j(x_{ij} + \ep w_{ij}) > f_j(x_{ij})$.

As before, let $X_j = \{x_{ij} \mid i \in N\} \setminus \{0\}$ be the set of distinct, non-zero amounts of good $j$ allocated to some agent under $\x$. Label these elements such that $\tau_j^1 < \tau_j^2 < \cdots  < \tau_j^{|X_j|}$. Since $f_j(0) = 0$, $f_j(x \notin X_j)$ in some sense doesn't matter -- we only require that these ``in-between'' areas of the price curve are weakly increasing and don't violate continuity. Thus there exist price curves $\f$ supporting $\x$ if and only if there exist functions $f_j': X_j \to \bbrpos$ such that 
\begin{enumerate}
\item for all $j \in M$, $0 \leq f_j'(\tau_j^1) \leq f_j'(\tau_j^2) \leq \cdots \leq f_j'(\tau_j^{|X_j|})$ (weakly increasing)
\item for all $i \in N$, $C_\f(x_i) = \sum_{j} f_j'(x_{ij}) = 1$ (total cost 1)
\item for all $i \in N$, exists $r,j \in M$ such that $f_j'(\tau_j^r = x_{ij} \neq 0) < f_j'(\tau_j^{r+1})$ (positive marginal cost)
\end{enumerate}

Now we are ready to set up the matrices $A,B,C$ (all of width $\sum_j |X_j| + 1$) to which we will apply Lemma~\ref{lem:transposition}. As in the proof of Theorem~\ref{thm:gdf-final}, $A$ will be the agent-order matrix, and the solution vector $\y$ will represent the marginal prices, with the last entry representing the total cost per agent. Thus, define
$$A\left[i, \ \sum_{\ell < j} |X_\ell| + q\right] = 
\begin{cases}
-1 &\text{if } j=m+1, q=1 \text{ (last column)} \\
\phantom{-}0 &\text{if } x_{ij} < \tau_j^{q} \\
\phantom{-}1 &\text{otherwise}
\end{cases}$$
\begin{figure}
\centering
\begin{subfigure}[b]{.45\textwidth}
\centering
{\setlength{\tabcolsep}{6pt}
\begin{tabular}{cccc:c:cc:ccc}
\ldelim[{4}{15pt}\hspace{-14pt} & 1 & 1 & 1 & 0 & 1 & 0 & -1 & \hspace{-10pt}\rdelim]{4}{15pt} \\
& 1 & 1 & 0 & 0 & 0 & 0 & -1 & \\
& 1 & 0 & 0 & 1 & 1 & 1 & -1 & \\
& 0 & 0 & 0 & 0 & 1 & 1 & -1 & \\
\end{tabular}
}
\caption{$\x$ represented as a agent-order matrix $A$}
\end{subfigure}
\hfill
\begin{subfigure}[b]{.45\textwidth}
\centering
{\setlength{\tabcolsep}{6pt}
\begin{tabular}{cccc:c:cc:ccc}
\ldelim[{4}{15pt}\hspace{-14pt} & 1 & 1 & 1 & 1 & 1 & 1 & 1 & \hspace{-10pt}\rdelim]{4}{15pt} \\
& 0 & 0 & 1 & 0 & 0 & 0 & 0 & \\
& 1 & 1 & 1 & 1 & 1 & 1 & 1 & \\
& 1 & 1 & 1 & 1 & 1 & 1 & 1 & \\
\end{tabular}
}
\caption{the corresponding marginal-cost matrix $C$}
\end{subfigure}
\caption{Example construction of the marginal-cost matrix from an agent-order matrix.}\label{fig:laf_intuition}
\end{figure}
Furthermore, let $B$ be the square identity matrix $I$; this will ensure that the prices are weakly increasing. Finally, we need to define the marginal-cost matrix $C$. As shown in Figure~\ref{fig:laf_intuition}, we can create $C$ based only on $A$: If agent $i$ receives the largest amount of some good (row $i$ has a 1 in the last column of some sub-block), then agent $i$'s row in $C$ is all 1's. Intuitively, we can set the price above $\max_i x_{ij}$ arbitrarily to ensure $i$ has positive marginal cost, so it should be trivial to satisfy $C_i \y > 0$. Otherwise, agent $i$'s row is all zeros, except that within a sub-block if there is a 1 followed by a 0 in row $i$ in $A$, the position of that 0 becomes a 1 in $C$. Intuitively, these are the places $i$ would have to buy more of a good to increase her utility. Formally, define

$$C\left[i, \ \sum_{\ell < j} |X_\ell| + q\right] = 
\begin{cases}
1 &\text{if } \exists j' \ A\left[i, \sum_{\ell \leq j'} |X_\ell|\right] = 1 \\
1 &\text{if } q \geq 1, A\left[i, \ \sum_{\ell < j} |X_\ell| + q\right] = 0, A\left[i, \ \sum_{\ell < j} |X_\ell| + q - 1\right]=1\\
0 &\text{otherwise}
\end{cases}$$

Since $\x$ gives at least one agent a nonempty bundle by assumption, $A,B,C$ have at least two columns. We know by Lemma~\ref{lem:transposition} that $\exists \y$ such that $A \y = \mathbf{0}, B \y \geq \mathbf{0}, C \y \gg \mathbf{0}$ if and only if $\not \exists \mathbf{u}, \mathbf{v}, \mathbf{w}$ such that $A^T \mathbf{u} + B^T \mathbf{v} + C^T\mathbf{w} = \mathbf{0}, \mathbf{v} \geq \mathbf{0}, \mathbf{w} > \mathbf{0}$. To complete the proof, we will show that the former condition is equivalent to the existence of weakly increasing price curves supporting $\x$, and the latter is equivalent to either $\x$ or an agent $i$ being locked.

Define $f_j'(\tau_j^q) - f_j'(\tau_j^{q-1}) = y_{\sum_{\ell < j} |X_\ell| + q}$, where for convenience we let $f_j'(\tau_j^0) = f_j'(0) = 0$. Clearly $B \y = \y \geq \mathbf{0}$ is equivalent to the requirement that price curves be weakly increasing. Furthermore, note that $C \y \gg \mathbf{0}$ implies $\y > 0$, so without loss of generality we can assume the last entry of $\y$ is 1. Thus as before, $A \y = 0$ is equivalent to the requirement that every agent's total cost equals 1. Revisiting $C \y \gg \mathbf{0}$, since $\y > \mathbf{0}$ this is trivially satisfied for every row where agent $i$ receives the largest amount of some good -- equivalently, agent $i$'s marginal cost can trivially be made positive. Additionally, for all other agents, $C_i \y > 0$ is by definition equivalent to having positive marginal cost. Thus a solution vector $\y$ is equivalent to weakly increasing price curves supporting $\x$.

If no such solution exists, then we have $A^T \mathbf{u} + B^T \mathbf{v} + C^T\mathbf{w} = \mathbf{0}, \mathbf{v} \geq \mathbf{0},\mathbf{w} > \mathbf{0}$. Rearranging, and since $B = I$, this is equivalent to $A^T \mathbf{u} \geq C^T \mathbf{w}, \mathbf{w} > \mathbf{0}$. Without loss of generality, assume $\mathbf{w}$ is only non-zero on entry $i$. Furthermore, for all $k$ define $a_k = u_k$ if $u_k > 0$ and $b_k = - u_k$ if $u_k < 0$. Then $A^T \mathbf{u} \geq C^T \mathbf{w}$ is equivalent to $\A \succ \B$ such that the domination is strict wherever $C_i$ is non-zero. If $C_i = \mathbf{1}$, this is equivalent to allocation $\x$ being locked. Otherwise, this is equivalent to agent $i$ being locked. Thus $A^T \mathbf{u} \geq C^T \mathbf{w}, \mathbf{w} > \mathbf{0}$ is equivalent to something being locked in $\x$.
\end{proof}

Finally, we observe that LAF give us the following linear program, which computes price curves (or shows that none exist) in polynomial time.

\begin{theorem}\label{thm:laf-lp}
Given a set of agents $N$, goods $M$, and an allocation $\x \in \bbr_{\ge 0}^{n \times m}$, let $A$ be the corresponding agent-order matrix and $C$ the marginal-cost matrix. In the following linear program, the optimal objective value is strictly positive if and only if there exist strictly increasing price curves supporting $\x$, in which case $\y$ defines such curves.
\begin{align*}
\max_{\y, \eta}\  \eta \quad \quad &\ \\
s.t. \quad A\y =&\ \mathbf{0}\\
	y_k \ge&\ 0\quad \forall k\\
    C_i \y \geq&\ \eta\quad \forall i\\
	y_{-1} =&\ 1
\end{align*}
\end{theorem}
\begin{proof}
As per the proof of Theorem~\ref{thm:laf}, there exist strictly increasing price curves supporting $\x$ if and only if there is a solution to the system $A \y = \mathbf{0}, \y \geq 0, C \y \gg \mathbf{0}$. To turn this into a valid linear program, we replace the strict inequality $C_i \y > 0$ with $C_i \y \geq \eta$ and attempt to maximize $\eta$. Furthermore, we restrict the final entry of $\y$ as $\y_{-1} = 1$, since otherwise $\y$ can be scaled arbitrarily. If there is a solution with $\eta > 0$, then $\y$ corresponds to price curves as before, with each entry representing the difference in price between adjacent allocation amounts. These points simply need to be connected, e.g., in a piecewise linear fashion, to constitute valid price curves.
\end{proof}

\section{Counterexamples}\label{sec:counterexamples}

\begin{example}[H]

\centering
\begin{tabular}{ c|cc} 
& agent 1 & agent 2\\
\hline
 good 1 & 1 & 1  \\
good 2 & 1 & 0
\end{tabular}
\caption{An instance where it is necessary to give a price of zero to some goods (which is a form of weakly increasing price curves) in order to support the maximum Nash or CES welfare allocation. Assume each good has supply 1. Nash welfare is maximized by splitting good 1 evenly between the two agents, and allowing agent 1 to purchase an equal quantity of good 2. This only possible if the price of good 2 is zero: otherwise, agent 1 is paying more than agent 2. It can be verified that this same allocation is also the maximum CES welfare allocation for any $\rho \in (-\infty, 0)\cup(0,1)$. For another interpretation, recall that the Fisher market equilibrium prices are the dual variables of the convex program for maximizing Nash welfare: thus the price of good 2 being zero corresponds to the fact that the supply constraint for good 2 is not tight in this instance.}
\label{ex:need-zero-prices}
\end{example}

%

We showed in Section~\ref{sec:ces-duality} that if $w_{ij} \in \{0,1\}$ for all $i \in N$ and $j \in M$, then for any $\rho \in (-\infty, 0) \cup (0,1)$, every maximum CES welfare allocation can be supported by price curves. One natural question is whether this result holds if we only assume that $\max_{j \in M} w_{ij} = 1$ for all $i \in N$. The answer is no, unfortunately, as demonstrated by the following theorem. Theorem~\ref{thm:L-infty-counter} only rules out $\rho$ in the range $(\frac{1}{2}, 1)$, but we conjecture that counterexamples exist for all $\rho \in (-\infty, 0) \cup (0,1)$.

\begin{theorem}\label{thm:L-infty-counter}
For agents with Leontief utilities where $\max_{j \in M} w_{ij} = 1$ for all $i \in N$, for every $\rho \in (\frac{1}{2}, 1)$, there exist instances where no maximum CES welfare allocation can be supported by price curves.
\end{theorem}

\begin{proof}
Consider the following instance with two goods with supply 1, and three agents, whose weights are given by the following table:

\begin{center}
\begin{tabular}{c|cc} 
& good 1 & good 2\\
\hline
agent 1 & $1- \ep$ & 1   \\
agent 2 & 1 & $1 - \ep$ \\
agent 3 & 1 & 1  \\
\end{tabular}
\end{center}

Let $\x$ be a maximum CES welfare allocation. For brevity, we write $u_i = u_i(x_i)$. In the proof of Theorem~\ref{thm:bandwidth} given in Section~\ref{sec:ces-duality}, we used duality to show that for a fixed $\rho$, any maximum CES welfare allocation $\x$ has the form
\[
x_{ij} = w_{ij} \Big( \sum_{j \in M} q_jw_{ij} \Big)^{\frac{1}{\rho -1}}
\]
for some constants $q_1,\dots, q_m \in \bbrpos$. Let $\chi_i = \Big( \sum_{j \in M} q_jw_{ij} \Big)^{\frac{1}{\rho -1}}$. In our case, we have
\begin{align*}
\chi_1 =&\ \big((1-\ep)q_1 + q_2\big)^{\frac{1}{\rho -1}}\\
\chi_2=&\ \big(q_1 + (1-\ep) q_2\big)^{\frac{1}{\rho -1}}\\
\chi_3 =&\ \big(q_1 +  q_2\big)^{\frac{1}{\rho -1}}
\end{align*}
Thus we have $x_{ij} = w_{ij} \chi_i$ for all $i \in N$ and $j \in M$. We proceed by case analysis.

Case 1: $(1-\ep) \chi_1 > \chi_3$. In this case, we have
\[
x_{11} = (1-\ep)\chi_1 > \chi_3 = x_{31} \quad\quad \text{and}\quad\quad x_{12} = \chi_1 > \chi_3 = x_{32}
\]
So $x_{1j} > x_{3j}$ for every good $j$. Let $\A$ be the vector with $a_1 = 1$ and $a_i = 0$ for $i\neq 1$, and let $\B$ be the vector with $b_3 = 1$ and $b_i = 0$ for $i \ne 3$. Then $\A\succ\B$. Furthermore: the domination is strict at $x_{3j}$ for each good $j \in M$. This means that agent 3 is locked. Therefore by Theorem~\ref{thm:laf}, the $\x$ cannot be supported by price curves, and we are done.

Case 2: $(1-\ep)\chi_2 > \chi_1$. By a symmetrical argument, we have $x_{2j} > x_{3j}$ for every good $j$, so agent 3 is again locked, and we are done. 

Case 3: $(1-\ep) \chi_1 \le \chi_3$ and $(1-\ep)\chi_2 \le \chi_3$. This implies that $(1-\ep)^{\rho -1} \chi_1^{\rho -1} \ge \chi_3^{\rho -1}$ and $(1-\ep)^{\rho -1} \chi_2^{\rho -1} \ge \chi_3^{\rho -1}$. Note that the inequality flipped because $\rho - 1 < 0$. Therefore
\begin{align*}
(1-\ep)^{\rho -1} \chi_1^{\rho -1} + (1-\ep)^{\rho -1} \chi_2^{\rho -1} \ge&\  2\chi_3^{\rho -1}\\
(1-\ep)^{\rho -1} \Big((1-\ep)q_1 + q_2\Big) + (1-\ep)^{\rho -1} \Big(q_1 + (1-\ep) q_2\Big) \ge&\  2\Big(q_1 +  q_2\Big)\\
(1-\ep)^{\rho -1} (2-\ep)\Big(q_1 +q_2\Big) \ge&\  2\Big(q_1 +  q_2\Big)\\
\ln\Big((1-\ep)^{\rho -1} (2-\ep)\Big) \ge&\  \ln 2\\
(\rho - 1) \ln (1-\ep)\ge&\  \ln 2 - \ln (2-\ep)\\
\rho \le&\  1+\frac{\ln 2 - \ln (2-\ep)}{ \ln (1-\ep)}\\
\end{align*}
Note that the sign flipped in the last step because $\ln(1-\ep) < 0$.

The resulting right hand side is some real-numbered value, so whenever $\rho$ is greater than that, we obtain a contradiction. Taking the limit as $\ep$ goes to 0 shows us that the right hand side may be arbitrarily close to $\frac{1}{2}$. This shows that for any $\rho > \frac{1}{2}$, there exists an $\ep > 0$ such that in the above instance, no maximum CES welfare allocation can be supported by price curves.
\end{proof}

\subsection{Difficulties in analyzing linear utilities}\label{sec:linear-hard}

We assumed throughout the paper that agents have Leontief utilities. One natural question is whether our results extend to other classes of utilities: in particular, linear utilities. The answer is no, in general. A \emph{linear} utility function is defined by
\[
u_i(x_i) = \sum_{j \in M} w_{ij} x_{ij}
\]
where $w_{ij}$ is still the weight that agent $i$ has for good $j$.

Leontief utilities have the very nice property that agents always purchase goods in a fixed proportion. It does not matter exactly how the cost within each bundle was distributed across goods, because each agent will always purchase goods in the same proportions, regardless of the underlying costs. We do not have this luxury with linear utilities. In this setting, the proportions in which each agent purchases goods depend on a complex interaction between her values for the goods, and the price curves. This makes it very difficult to reason about what agents will purchase given a set of price curves. In fact, each agent's optimization problem
\[
\argmax_{x_i \in \bbrpos^m:\ C_\f(x_i) \leq 1} u_i(x_i)
\]
may not even be convex.

Thus in order for $(\x, \f)$ to form a price curve equilibrium for linear utilities, a complex set of conditions would need to be satisfied. We note that $C_\f(x_i) = 1$ is still necessary, and so GDF is still a necessary condition (for strictly increasing price curves), but it is certainly not sufficient (Example~\ref{counter:ef-gdf}).

\begin{example}[tb]

\centering
\begin{tabular}{ c|cc} 
& agent 1 & agent 2\\
\hline
 good 1 & 4 & 0  \\
 good 2 & 0 & 4  \\
 good 3 & 1 & 2  \\
 good 4 & 2 & 1
\end{tabular}
\caption{For two agents with linear utilities, group-domination-freeness is not sufficient for the existence of price curves. Consider the instance where the agents' weights are given as above and the available supply of each good is 1. Define $\x$ by $x_{11} = x_{13} = x_{22} = x_{24} = 1$ and $x_{ij} = 0$ otherwise. This allocation is EF and GDF. To see that $\x$ cannot be supported by price curves, let $\tilde{j} = \argmin_{j \in \{3,4\}} f_j(1)$. If $\tilde{j} = 3$, then the cost of good 3 is at most the cost of good 4, so agent 2 would buy good 3 instead of buying good 4. Similarly, if $\tilde{j} = 4$, then the cost of good 4 is at most the cost of good 3, so agent 1 would buy good 4 that instead of buying good 3.}
\label{counter:ef-gdf}
\end{example}

\section{Omitted proofs}\label{sec:extra-proofs}

\lemLeontiefDemand*

\begin{proof}
$(\impliedby)$ Suppose the above conditions hold, but $x_i \not\in D_i(\f)$. Then there exists $x_i'\in D_i(\f)$ such that $u_i(x_i') = u_i(x_i) + \ep$ for some $\ep > 0$. Since we assume that $x_{ij}$ is proportional to $w_{ij}$, agent $x$ must receive at least $\ep w_{ij}$ more of each good $j$ in order to increase her utility by $\ep$. Furthermore, since price curves are increasing, $f_j(x_{ij}') \geq f_j(x_{ij})$ for every good $j$. However, condition (b) of the lemma implies that there exists a good $j$ such that
\[f_j(x_{ij}') \geq f_j(x_{ij} + \ep w_{ij}) > f_j(x_{ij})\]
and thus
\[C_\f(x_i') =\ \sum_{j \in M} f_j(x_{ij}') > \sum_{j\in M} f_j(x_{ij}) = 1\]
which contradicts $x_i' \in D_i(\f)$.

$(\implies)$ Now suppose that at least one of the two conditions of the lemma does not hold. If $C_\f(x_i) \neq 1$, then either $C_\f(x_i) > 1$ and the cost exceeds the budget, or $C_\f(x_i) < 1$ so by continuity agent $i$ could purchase more of every good and increase her utility. Either way $x_i \not\in D_i(\f)$. Thus assume that for every $j$, there exists an $\ep_j > 0$ such that $f_j(x_{ij} + \ep_j w_{ij}) = f_j(x_{ij})$. Then consider the bundle $x_i'$ defined by $x_{ij}' = x_{ij} + \ep_j w_{ij}$. This bundle has the same cost as $x_i$, but
\[u_i(x_i') = \min_{j \in M} \frac{x_{ij} + \ep_j w_{ij}}{w_{ij}} > \min_{j \in M} \frac{x_{ij}}{w_{ij}} = u_i(x_i)\]
contradicting $x_{i} \in D_i(\f)$.
\end{proof}


\bandwidthConverse*

\begin{proof}
First, for Nash welfare ($\rho = 0$), this is exactly Eisenberg and Gale's result: the linear-pricing equilibrium allocations are exactly the allocations maximizing Nash welfare~\cite{eisenberg_aggregation_1961, eisenberg_consensus_1959}. Thus for the rest of this proof, we assume $\rho \ne 0$.

The proof follows a duality argument very similar to the proof of Theorem~\ref{thm:bandwidth}. We use the same convex program for maximizing CES welfare, which, as stated in the proof of Theorem~\ref{thm:bandwidth}, satisfies strong duality. Suppose $(\xs, \g)$ is a PCE, where $g_j(x) = q_j^* x^{1-\rho}$ for all $j \in M$ for nonnegative constants $q_1^*\dots q_m^*$. Let $u_i^* = u_i(x_i^*)$ be agent $i$'s utility for $\xs$, and let $\lambda_{ij}^* = q_j^*$. Let $\mathbf{u^*} = u_1^*\dots u_n^*$, let $\mathbf{q^*} = q_1^*\dots q^*_m$, and let $\boldsymbol{\lam^*}$ represent the vector of all $\lambda_{ij}^*$'s.

Since our convex program satisfies strong duality, the \emph{Karush-Kuhn-Tucker} (KKT) conditions are both necessary and sufficient for optimality. Specifically, if we can show that $(\xs, \mathbf{u^*}, \boldsymbol{\lam^*}, \mathbf{q^*})$ satisfies the KKT conditions, then $(\xs, \mathbf{u^*})$ is optimal for the primal. The KKT conditions are primal feasibility, dual feasibility, complementary slackness, and stationarity. Since $\xs$ is a valid allocation and $u_i^*$ is defined by $u_i^* = u_i(x_i^*)$ for all $i \in N$, primal feasibility of $(\xs, \mathbf{u^*})$ immediately follows. Since $q_j^* \ge 0$ for all $j \in M$ by assumption and $\lambda_{ij}^* \ge 0$ by definition, we have dual feasibility as well.

Complementary slackness requires that for every constraint, either the constraint is tight, or the corresponding dual variable is equal to 0. For the supply constraints, we need to show that for all $j \in M$, we have $\sum_{j \in M} x_{ij}^* = s_j$ whenever $q_j^* = 0$. This is satisfied by assumption. For the other constraints, we need to show that for all $i \in N, j \in R_i$, either $\lambda_{ij}^* = 0$ or $x_{ij}^* = u_i^*$. We will show something slightly stronger: either $\lambda_{ij}^* = 0$ or $x_{ij}^* = w_{ij} u_i^*$ (for all $j \in M$, not just in $R_i$). Since $\lambda_{ij}^* = q_j^*$, we have $\lambda_{ij}^* = 0$ when $q_j^* = 0$. Suppose $q_j^* \ne 0$ and $x_{ij}^* \ne w_{ij} u_i^*$. First, we must have $x_{ij}^* \ge w_{ij} u_i^*$ by the definition of Leontief utility, which implies that $x_{ij}^* > w_{ij} u_i^*$. Also, since $q_j^* \ne 0$, agent $i$ must be spending money on good $j$; furthermore, she is purchasing more of good $j$ than she needs. Instead, she could purchase $x'_{ij} = w_{ij} u_i^*$ and have some leftover money, which she could use to buy slightly more of every good and increase her utility. This would imply that $x_i^*$ is not in agent $i$'s demand set, which contradicts $(\xs, \f)$ being a price curve equilibrium. Therefore we have $x_{ij}^* = w_{ij} u_i^*$ whenever $q_j \ne 0$, which satisfies the complementary slackness conditions.

For stationarity, we need to show that the gradient of $L$ with respect to $\x$ and $\bfu$ vanishes at $(\mathbf{x^*}, \mathbf{u^*}, \boldsymbol{\lam^*}, \mathbf{q^*})$ for every coordinate that is not zero. Specifically, we need to show that for each variable $y$, either  $\frac{\partial L}{\partial y} = 0$, or $y = 0$ and $\frac{\partial L}{\partial y} \le 0$. First, for $j \in R_i$ we have $\mfrac{\partial L}{\partial x_{ij}} (\mathbf{x^*}, \mathbf{u^*}, \boldsymbol{\lam^*}, \mathbf{q}^* = \lam_{ij}^* - q_j^*$; by definition, this is equal to zero. For $j \not\in R_i$, we have $\mfrac{\partial L}{\partial x_{ij}} (\mathbf{x^*}, \mathbf{u^*}, \boldsymbol{\lam^*}, \mathbf{q}^*) = - q_j^*$. This is always nonpositive, since $q_j^* \ge 0$. We showed before that $x_{ij}^* = w_{ij} u_i^*$ whenever $q_j \ne 0$. Since $w_{ij} = 0$ for $j \not\in R_i$, we have either $x_{ij}^* = 0$ or $q_j^* = 0$; this satisfies the stationarity condition for those variables.

Finally, consider $u_i^*$: since $(\xs, \f)$ is a price curve equilibrium, everyone must have positive utility (any agent could always buy a nonzero amount of every good, which would give her nonzero utility). Thus we need to show that $\mfrac{\partial L}{\partial u_i} (\mathbf{x^*}, \mathbf{u^*}, \boldsymbol{\lam^*}, \mathbf{q}^*) = {u_i^*}^{\rho - 1} - \sum_{j \in R_i} \lam_{ij}^* = 0$ for each $i \in N$. Since $(\xs, \f)$ is a price curve equilibrium, each agent must be exhausting her entire budget. Thus $C_\f(x_i^*) = 1$ for all $i \in N$, which gives us:
\[
C_\f(x_i^*) = \sum_{j \in M} f_j(x_{ij}^*) = \sum_{j \in M} q_j^* {x_{ij}^*}^{1-\rho} = 1
\]
Furthermore, as argued above, we have $x_{ij}^* = w_{ij} u_i^*$ whenever $q_j \ne 0$. Therefore
\begin{align*}
\sum_{j \in M} q_j^* {x_{ij}^*}^{1-\rho} =&\ \sum_{j: q_j^* \ne 0} q_j^* {x_{ij}^*}^{1-\rho}\\
=&\ \sum_{j: q_j^* \ne 0} q_j^* (w_{ij} u_i^*)^{1-\rho}\\
=&\ {u_i^*}^{1-\rho}\sum_{j: q_j^* \ne 0} q_j^* w_{ij}^{1-\rho} \\
=&\ {u_i^*}^{1-\rho}\sum_{j: q_j^* \ne 0} q_j^* w_{ij}
\end{align*}
where the last equality is because $w_{ij} \in \{0,1\}$. Therefore we have \begin{align*}
{u_i^*}^{1-\rho}\sum_{j: q_j^* \ne 0} q_j^* w_{ij} =&\ 1\\
\sum_{j: q_j^* \ne 0} q_j^* w_{ij} =&\ {u_i^*}^{\rho - 1}\\
{u_i^*}^{\rho - 1} - \sum_{j: q_j^* \ne 0} q_j^* w_{ij} =&\ 0 \\
{u_i^*}^{\rho - 1} - \sum_{j \in R_i} q_j^* =&\ 0\\
{u_i^*}^{\rho - 1} - \sum_{j \in R_i} \lam_{ij}^* =&\ 0
\end{align*}
Therefore $\mfrac{\partial L}{\partial u_i} (\mathbf{x^*}, \mathbf{u^*}, \boldsymbol{\lam^*}, \mathbf{q}^*) = {u_i^*}^{\rho - 1} - \sum_{j \in R_i} \lam_{ij}^*$ is indeed 0. Thus the KKT conditions are satisfied. Therefore $(\xs, \mathbf{u^*}, \boldsymbol{\lam^*}, \mathbf{q^*})$ is optimal for $L$, which implies that $(\mathbf{x^*}, \mathbf{u^*})$ is optimal for the primal: in other words, $\xs$ is a maximum CES welfare allocation.
\end{proof}

\end{document}